\documentclass[journal,10pt,twocolumn]{IEEEtran}

\IEEEoverridecommandlockouts 
\usepackage{cite}
\usepackage{subfigure}
\usepackage{amsmath,amssymb,amsfonts,stmaryrd}
\usepackage{amsthm,mathtools}
\usepackage{algorithmic}
\usepackage{algorithm}
\usepackage{comment}
\usepackage{graphicx}
\usepackage{textcomp}
\usepackage{xcolor}

\usepackage{xparse} 
\DeclarePairedDelimiterX{\Iintv}[1]{\llbracket}{\rrbracket}{\iintvargs{#1}}
\NewDocumentCommand{\iintvargs}{>{\SplitArgument{1}{,}}m}
{\iintvargsaux#1} %
\NewDocumentCommand{\iintvargsaux}{mm} {#1\mkern1.5mu..\mkern1.5mu#2}
\newtheorem{remark}{Remark}

\newtheorem{theorem}{Theorem}
\newtheorem{lemma}{Lemma}

\newtheorem{assumption}{Assumption}
\newtheorem{corollary}{Corollary}
\newtheorem{definition}{Definition}

\makeatletter
\renewcommand{\subsection}{\@startsection{subsection}{2}{\z@}%
  {0.6ex plus 0.3ex minus 0.1ex}%
  {0.4ex plus 0.2ex}%
  {\normalfont\normalsize\itshape}}
\makeatother

\def\BibTeX{{\rm B\kern-.05em{\sc i\kern-.025em b}\kern-.08em
    T\kern-.1667em\lower.7ex\hbox{E}\kern-.125emX}}

\begin{document}
\title{Models, Methods and Waveforms for Estimation and Prediction of Sparse Time-Varying Channels}
\author{
Wissal Benzine\textit{$^{1,2}$}, Ali Bemani\textit{$^1$}, Nassar Ksairi\textit{$^1$}, and Dirk Slock\textit{$^2$} \\
\textit{$^1$}Mathematical and Algorithmic Sciences Lab, Huawei France R\&D, Paris, France \\ \textit{$^2$}Communication Systems Department, EURECOM, Sophia Antipolis, France\\
Emails:
\{Wissal.Benzine, Dirk.Slock\}@eurecom.fr, \{ali.bemani, nassar.ksairi\}@huawei.com
}
\maketitle

\begin{abstract}
This paper investigates channel estimation for linear time-varying (LTV) wireless channels under \emph{double sparsity}, i.e., sparsity in both the delay and Doppler domains. An on-grid approximation is first considered, enabling rigorous hierarchical-sparsity modeling and  compressed sensing-based channel estimation. Guaranteed recovery conditions are provided for affine frequency division multiplexing (AFDM), orthogonal frequency division multiplexing (OFDM) and single-carrier modulation (SCM), highlighting the superiority of AFDM in terms of doubly sparse channel estimation.
To address arbitrary Doppler shifts, a relaxed version of the on-grid model is introduced by utilizing multiple elementary Expansion Models (BEM) each based on Discrete Prolate Spheroidal Sequences (DPSS). Next, theoretical guarantees are provided for the precision of this off-grid model before further extending it to tackle channel prediction by exploiting the inherent DPSS extrapolation capability. Finally, numerical results are provided to both validate the proposed off-grid model for channel estimation and prediction purposes under the double sparsity assumption and to compare the corresponding mean squared error (MSE) and overhead performance when different wireless waveforms are used.
\end{abstract}


\section{Introduction}

Sparsity is an important feature of wireless channels, particularly in high-frequency bands. Sparsity emerges due to the limited number of dominant scatterers contributing to signal propagation and manifests in both the delay and Doppler domains. By leveraging sparsity, advanced estimation methods can mitigate the challenges posed by high-mobility scenarios on channel estimation and prediction. The question thus arises as to 1) what delay-Doppler modeling can best leverage sparsity to enable efficient and low-complexity estimation and prediction of time-varying channels, and 2) which wireless waveforms can make the best use of delay-Doppler sparsity e.g., for pilot overhead reduction.

In wireless communication systems operating in low-to-moderate mobility scenarios, sparsity is often observed in the delay domain, where only a few dominant delay taps carry most of the channel energy \cite{Non_gaussian_MP}. Sparsity extends to the Doppler domain as well in high-mobility scenarios such as vehicular or high-speed train communications. In both mobility regimes, sparsity is more accentuated in the beamspace domain of MIMO channels, especially when the number of antennas is large \cite{sayeed2002,sayeed2011,beamspace}.
Channel estimation methods leveraging delay-domain sparsity have been proposed in various studies, such as \cite{ofdm_ongrid_delay} in which a grid-based discretization of the delay dimension is used to enable a compressive sensing framework. Delay-Doppler sparsity was assumed in \cite{kron_sbl} and leveraged to conceive enhanced channel estimation schemes for time-varying channels. However, sparsity was modeled as the sparsity of a one-dimensional array, with no way to assign different sparsity levels to the delay and Doppler domains. The model in \cite{ofdm_dD_sparsity} also assumes a form of delay-Doppler sparsity, though under the restrictive assumption of one Doppler shift per delay tap.
On-grid models, assuming that delay and Doppler shifts are quantized to a predefined grid, allow for structured sparse recovery methods. This quantization simplifies estimation but introduces grid mismatch errors \cite{chi2011mismatch}. To overcome the limitations of strict grid-based models, off-grid channel estimation techniques have been introduced by relaxing the assumption of discretized delay or Doppler shifts. Having estimates of the fractional part of Doppler frequency shifts enables straightforward channel extrapolation and, hence, prediction. Grid refinement techniques \cite{Ganguly2019refinedgrid,Tian2016MicrowaveSC} have been explored as a simple and low-complexity method to mitigate the performance degradation caused by grid mismatches. More advanced gridless approaches offering super-resolution recovery have also been investigated \cite{beinert2021superresolution,sparse_sbl,kron_sbl,distributedCS_offgrid}. Off-grid sparse Bayesian learning (SBL) \cite{sparse_sbl,kron_sbl} is one such solution. Another solution based on atomic norm minimization is proposed in \cite{distributedCS_offgrid}. Such approaches can, in principle, offer better performance than grid refinement solutions, but at the cost of higher computational complexity. Moreover, they cannot fully eliminate the basis mismatch between real-world signals and the grid, leading to residual errors. They could also suffer from ill conditioning. This is particularly the case, as argued in Section \ref{sec:off_grid} and validated in Section \ref{sec:simulations}, when the propagation link is characterized by a large number of ``physical'' channel paths contributing to each of the (refined) grid points. As an alternative, basis expansion models (BEM) do not attempt to estimate the individual off-grid delay-Doppler components but rather their combined effect, thus avoiding the ill-conditioning issue. Moreover, they can, in principle, be made to exploit sparsity. For example, the BEM based methods in \cite{Mohebbi2021ANC,Fu2016CS,embedded_otfs_bem} take advantage of sparsity, though only in the delay domain. There is thus a need to add support for Doppler domain sparsity to the BEM approach. 

Concerning waveforms conceived to deal with high-mobility scenarios, orthogonal time frequency space (OTFS) uses a two-dimensional (2D) modulation technique that transmits the information and pilot symbols in the delay-Doppler domain \cite{otfs}. However, OTFS channel estimation overhead cannot be reduced when the channel exhibits sparsity \cite{otfs_cs}. Indeed, the performance gap between OTFS and OFDM narrows in favor of OFDM on channels with delay domain sparsity \cite{otfs_ofdm_sparsity}. Another recently proposed waveform for communications in high-mobility scenarios is Affine Frequency Division Multiplexing (AFDM) \cite{BemaniAFDM_TWC}. While in AFDM, data and pilot symbols are not directly transmitted in the delay-Doppler domain; AFDM can still reconstruct a delay-Doppler representation of the channel to achieve full diversity on doubly dispersive channels. In the absence of sparsity, this leads to AFDM having a comparable bit error rate (BER) performance to that of OTFS, but with the advantage of requiring less channel estimation overhead \cite{BemaniAFDM_TWC}. Under sparsity, our previous works \cite{afdm_gc,afdm_spawc} proved the superiority of AFDM over both OFDM and OTFS in terms of pilot overhead using an on-grid channel model. In the current work, we expand our results to time-varying channels with off-grid Doppler shifts and we address channel prediction as a key application of this extension.

\subsection*{Contributions}
\begin{enumerate}
    \item The statistical notion of delay-Doppler double sparsity (DS) for LTV channels is defined and rigorously linked to the hierarchical sparsity paradigm under a first on-grid approximation of Doppler shifts . This link allows us to use hierarchical-sparsity mathematical tools from the compressed-sensing paradigm to analyze the effect of the used waveform on DS-LTV channel estimation, pointing to the superiority of AFDM over other waveforms in terms of sparse channel estimation overhead.
    \item The on-grid model is relaxed to allow for off-grid Doppler shifts, resulting in a novel DS-LTV channel representation with theoretical precision guarantees based on multiple frequency-shifted elementary BEMs. This channel representation constitutes a novel approach to modeling delay-Doppler sparsity. In this approach, sparsity is not synonymous with relatively small numbers of channel paths but consists in assuming relatively small numbers of \emph{clusters} of such paths, with each cluster being associated with a delay-Doppler grid point.
    \item The proposed channel representation is then used to conceive a linear minimum mean square error (LMMSE) LTV channel estimator, replacing the more complex sparse-recovery approach that enabled us to obtain closed-form waveform performance analysis. The results of this theoretical waveform performance comparison, conducted under the on-grid approximation and the compressed-sensing paradigm, are shown using numerical results to continue to hold under the more realistic off-grid model and LMMSE channel estimation.
    \item A prediction method based on extrapolating the basis vectors of the above BEM representation is finally proposed. Its optimality is established by linking it to the reduced-rank MMSE predictor.
\end{enumerate}

\subsection*{Outlines}
Section \ref{sec:system_model} provides a general mathematical framework for modeling doubly sparse, time-varying wireless channels. Section \ref{sec:on_grid} is dedicated to a first on-grid approximation of that model and to the channel estimation performance of different waveforms under that approximation. Section \ref{sec:off_grid} introduces a relaxation that allows for fractional Doppler frequency shifts, along with a new model to capture the associated off-grid effects. Furthermore, the inherent properties of the proposed model are used to enable channel prediction. Finally, numerical results are provided in Section \ref{sec:simulations} to validate the proposed models and the DS-LTV channel estimation and extrapolation performance of different waveform candidates.

\subsection*{Notations}
$\mathrm{Bernoulli}(p)$ is the Bernoulli distribution with activation probability $p$ and $\mathrm{B}(n,p)$ is the binomial distribution with parameters $(n,p)$.
Notation $X\sim F$ means that the random variable $X$ has a distribution $F$.
If $\mathcal{A}$ is a set, $|\mathcal{A}|$ stands for its cardinality. The set of all integers between $l$ and $m$ for some $(l,m)\in\mathbb{Z}^2$ (including $l$ and $m$) is denoted $\Iintv{l,m}$. For a matrix $\mathbf{M}$, $[\mathbf{M}]_{c}$ stands for the $c$-th column.
The ceiling operation is denoted as $\lceil.\rceil$. The modulo $N$ operation is denoted as $(\cdot)_N$. $\mathcal{F}$ stands for the continuous-time Fourier transform operator, and $x\mapsto\mathrm{rect}(x)$ stands for the (standard) rectangular function. $\mathbf{I}_N$ is the $N\times N$ identity matrix.

\section{System Model: Linear time-varying (LTV) channels}
\label{sec:system_model}
A linear time-varying (LTV) channel is a model of multipath propagation that is characterized by changes in its impulse response over time, caused by Doppler frequency shifts. The received signal at the channel output corresponding to a signal $s(t)$ at its input is expressed as:
\begin{equation}
\label{eq:continuous_td_IO}
    r(t) = \int_{\tau} s(t)h(t,\tau)d\tau + z(t),t\in\mathbb{R},
\end{equation}
where $z(t)$ is the additive white Gaussian noise process and 
\begin{equation}
    h(t,\tau) = \sum _{p=1}^{N_{\rm p}} g_{p}e^{\imath2\pi \nu_p t\Delta f}\delta(\tau - \tau_pT_{\rm s}),\label{eq:continuous_channel_model}
\end{equation}
is the continuous-time impulse response of the channel.
Here, $N_{\rm p}\geq1$ is the number of paths, $\delta(\cdot)$ is the Dirac delta function, $\Delta f$ is the subcarrier spacing, $T_{\rm s}$ is the sample period, $g_p$ is the complex gain of the $p$-th path, $\nu_p\triangleq\frac{\text{Doppler frequency in Hz}}{\Delta f}$ and $\tau_p\triangleq\frac{\text{Delay in seconds}}{T_{\rm s}}$. Moreover, we write $\tau_p=l_p+\iota_p$ with $l_p$ as the integer part of $\tau_p$, while $\iota_p$ is the fractional part that satisfies $\frac{-1}{2} < \iota_i \leq \frac{1}{2}$. Finally, $\nu_p = q_p +\kappa_p $, where $q_p\in\Iintv{-Q,Q}$ is the integer part of $\nu_q$, while $\kappa_p$ is the fractional part satisfying $\frac{-1}{2} < \kappa_i \leq \frac{1}{2}$.
In practice, the transmitted signal $s(t)$ is the continuous-time version of a discrete-time signal $s_n\triangleq s\left(nT_{\rm s}\right)$ generated from a vector $\mathbf{x}$ of $N$ symbols. These symbols could be either data symbols, pilot symbols, or a combination of both. Defining $r_n \triangleq r(nT_{\rm s})$ and $z_n\triangleq z(nT_{\rm s})$, we obtain the discrete-time version of the model in \eqref{eq:continuous_td_IO}
\begin{equation}
\label{eq:discrete_td_IO}
    r_n = \sum_{p = 1}^{N_{\rm p}}g_pe^{\imath2\pi \nu_p n\Delta fT_{\rm s}}s(nT_{\rm s}-\tau_pT_{\rm s}) + z_n,\quad n\in\mathbb{Z}\:.
\end{equation}
From now on, the process $(z_n)_{n\in\mathbb{Z}}$ is modeled as independent and identically distributed (i.i.d.) with $z_n\sim\mathcal{CN}\left(0,\sigma_w^2\right)$.
  \vspace{-1mm}
  
\section{Estimation of doubly sparse LTV channels with the on-grid approximation}
\label{sec:on_grid}
\subsection{First approximation: on-grid doubly sparse LTV channels}
In a first approximation, we assume that both $\iota_p$ and $\kappa_p$ are zero, and we define  $L\triangleq \underset{p=1\cdots N_{\rm p}}{\max}\tau_p+1$.
The discrete-time input-output model in \eqref{eq:discrete_td_IO} becomes
\begin{equation}
\label{eq:on_grid_discrete_td_IO}
    r_n = \sum_{l = 0}^{L-1}s_{n-l}h_{l,n} + z_n,\quad n\in\mathbb{Z}\:.
\end{equation}
The input-output relation in \eqref{eq:on_grid_discrete_td_IO} defines an on-grid LTV channel with a $L-1$ maximum delay shift with the complex gain $h_{l,n}$ of the $l$-th path varying with the time index $n$ as
 \begin{equation}
     \label{eq:on_grid_ch_model}
     h_{l,n}=\sum _{p=1}^{N_{\rm p}} g_{p}e^{\imath2\pi \frac{q n}{N}}\delta(l - l_p), l=0\cdots L-1\:.
 \end{equation}
Under this assumption, we assume that $g_p = \alpha_{l_p, q_p} I_{l_p, q_p}$ where $I_{l, q}$ is given by:
\begin{equation}
\label{eq:indicator_function}
I_{l, q} = 
\begin{cases} 
1 & \text{if } \exists p \text{ such that } (l, q) = (l_p, q_p), \\
0 & \text{otherwise.}
\end{cases}
\end{equation}
Here, $I_{l,q}$ for any $l$ and $q$ is a binary random variable that, when non-zero, indicates that a channel path with delay $l$, Doppler shift $q$ and complex gain $\alpha_{l,q}$ is active and contributes to the channel output. The complex gain is assumed to satisfy the channel power normalization
\begin{equation}
    \label{eq:power_normalization_on_grid}
    \sum_{l=0}^{L-1}\sum_{q=-Q}^{Q}\mathbb{E}\left[\left|\alpha_{l,q}\right|^2I_{l,q}\right]=1\:.
\end{equation}
The number of paths $N_{\rm p}$, in \eqref{eq:continuous_channel_model} can now be expressed as
\begin{equation}
\label{Np_ongrid}
    N_{\rm p} = \sum_{l=0}^{L-1} \sum_{q=-Q}^{Q}I_{l,q}\:.
\end{equation}
Note that the distribution of the random variables $\left\{I_{l,q}\right\}_{l,q}$ controls the type of sparsity that the LTV channel exhibits. We opt for the distribution given by the following definition.
\begin{definition}[On-grid Delay-Doppler double sparsity, \cite{afdm_gc}]
\label{def:dD_sparsity}
The complex gain $h_{l,n}$ of the $l$-th path
varies with time as
\begin{equation}
    \label{eq:ch_model}
    h_{l,n} = \sum_{q=-Q}^{Q} \alpha_{l,q} I_{l,q} e^{\imath 2\pi \frac{nq}{N}}, \quad l = 0, \ldots, L-1,
\end{equation}
and there exist \(0 < p_d, p_D < 1\) such that
\begin{equation}
    I_{l,q} = I_l I_q^{(l)}, \quad \forall (l,q) \in \Iintv{0,L-1} \times \Iintv{-Q,Q},
\end{equation}
where \(I_l \sim \mathrm{Bernoulli}(p_d)\) and \(I_q^{(l)} \sim \mathrm{Bernoulli}(p_D)\). Moreover, $I_{l,q}$ and $\alpha_{l,q}$ are independent and $\alpha_{l,q}\sim\mathcal{CN}\left(0,\sigma_{\alpha}^2\right)$ with $\sigma_{\alpha}^2$ satisfying \eqref{eq:power_normalization_on_grid}.
\end{definition}
The use of the Bernoulli distribution to model channel sparsity is widely adopted \cite{bernoulli}. It is a natural way to probabilistically model the activation or lack of activation of a channel tap, i.e., the presence or absence of an actual channel path at that tap position. Indeed, the sparsity level of a collection of channel taps can be directly linked to the activation probability of the Bernoulli distribution modeling each of them. More precisely, note that under Definition \ref{def:dD_sparsity},

\begin{equation}
\label{eq:sd}
    s_{\rm d}\triangleq\mathbb{E}\left[\sum_{l}I_l\right]=p_dL
\end{equation}

is the mean number of active delay taps in the delay-Doppler profile of the channel and can be thought of as the {\it delay domain sparsity level} while

\begin{equation}
\label{eq:sD}
    s_{\rm D}\triangleq\mathbb{E}\left[\sum_{q}I_q^{(l)}\right]=p_D(2Q+1)
\end{equation}

is the mean number of active Doppler bins per delay tap and can be thought of as the {\it Doppler domain sparsity level}.
Fig. \ref{fig:examples} illustrates three different delay-Doppler sparsity models, fully described in \cite{afdm_gc} and dubbed Type-1, Type-2 and Type-3. All these models fall under the scope of Definition \ref{def:dD_sparsity}, each with an additional assumption on $I_l$ and $I_q^{(l)}$. For instance, in Type-3 models of Figure \ref{fig:examples}-(c), the active Doppler bins per delay tap appear in clusters of random positions but of deterministic length, as opposed to the absence of clusters in Type-2 models of Figure \ref{fig:examples}-(b). The case in which the delay taps have all the same (random) sparsity (as in Type-1 models of Fig. \ref{fig:examples}-(a)) also falls under Definition \ref{def:dD_sparsity} by setting $I_q^{(l)}=I_q^{(0)},\forall l$.
\begin{figure}
  \centering
  \begin{tabular}{ c @{\hspace{5pt}} c @{\hspace{5pt}} c}
  \includegraphics[width=.3\columnwidth] {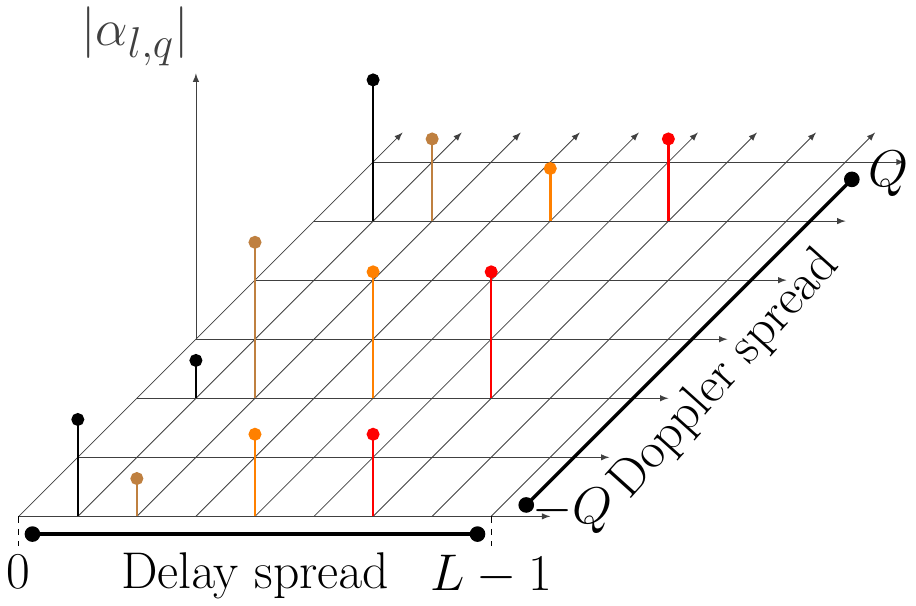} &
    \includegraphics[width=.3\columnwidth]{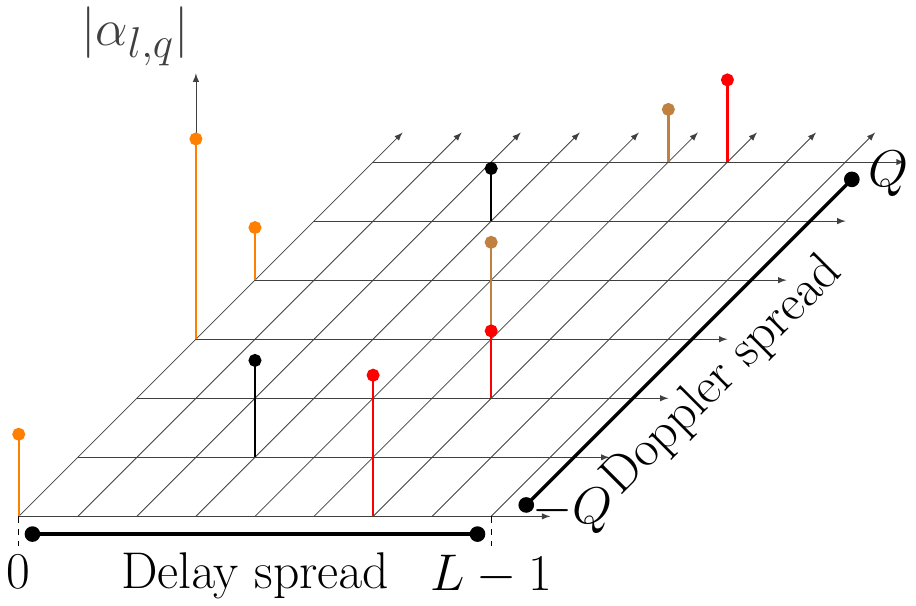} &
      \includegraphics[width=.3\columnwidth]{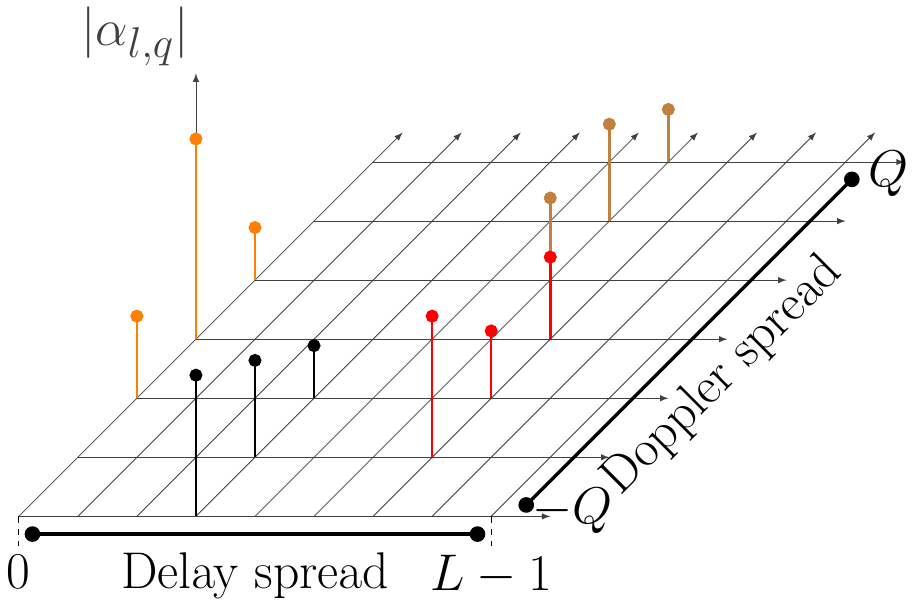} \\
    \small (a) &
      \small (b)&
      \small (c)
  \end{tabular}
  \medskip
  \caption{Examples of channels satisfying (a) Type-1, (b) Type-2, (c) Type-3 delay-Doppler sparsity}
  \label{fig:examples}
\end{figure}

\subsection{Relation to hierarchical sparsity \cite{hierarchical}}

\begin{definition}[Hierarchical sparsity]
\label{def:hierarchical}
Vector $\mathbf{x} \in \mathbb{C}^{NM}$ is $\left(s_{N},s_{M}\right)$-sparse if it consists of $N$ size-$M$ blocks, $s_{N}$ of which at most are non-zero, and each non-zero block is $s_{M}$-sparse.
\end{definition}
To analyze hierarchically sparse recovery schemes, a modified version of the restricted isometry property (RIP) called the hierarchical RIP (HiRIP) was proposed in the literature. 
\begin{definition}[HiRIP, \cite{hierarchical}]
\label{def:HiRIP}
The HiRIP constant $\delta_{s_{N},s_{M}}$ of a matrix $\mathbf{A}$ is the smallest $\delta \geq 0$ such that for all $\left(s_{N},s_{M}\right)$-sparse vectors $\mathbf{x} \in \mathbb{C}^{NM}$
\begin{equation}
     \label{eq:HiRIP}
     (1- \delta)\left\|\mathbf{x}\right\|^2\leq \left\|\mathbf{A}\mathbf{x}\right\|^2\leq(1+ \delta)\left\|\mathbf{x}\right\|^2
 \end{equation}
\end{definition}
DS-LTV sparsity is probabilistic while hierarchical sparsity of Definition \ref{def:hierarchical} is deterministic. The two models are nonetheless related: if vectorized to a concatenation of its rows, random matrix $\left[\alpha_{l,q}I_{l,q}\right]_{l,q}$ defines a vector $\boldsymbol{\alpha}\in\mathbb{C}^{(2Q+1)L}$ that consists of $L$ blocks each of size $2Q+1$, where on average $s_{\rm d}$ blocks have non-zero entries and where each non-zero block itself is on average $s_{\rm D}$-sparse. To ensure sparsity in a stronger sense, i.e., with high probability (as $L,Q,Lp_{\rm d},(2Q+1)p_{\rm D}$ grow), we require that the following assumption holds.
\begin{assumption}
    \label{assum:technical}
    $\left\{I_l\right\}_{l=0\cdots L-1}$ are mutually independent. Moreover, the complementary cumulative distribution function (CCDF) $\overline{F}_{S_{\mathrm{D},l}}(m)$ of the random variable $S_{\mathrm{D},l}\triangleq\sum_{q=-Q}^{Q}I_{q}^{(l)}$ for any $l\in\Iintv{0,L-1}$ is upper-bounded for any integer $m>(2Q+1)p_{\rm D}$ by the CCDF of $\mathrm{B}\left(2Q+1,p_{\rm D}\right)$.
\end{assumption}
Type-1 and 2 models satisfy the CCDF upper-bound condition simply by requiring that $\{I_{q}^{(0)}\}_q$ in the first and $\{I_{q}^{(l)}\}_q$ for any $l$ in the second be mutually independent (and thus satisfy $\overline{F}_{S_{\mathrm{D},l}}(m)=\overline{F}_{\mathrm{B}\left(2Q+1,p_{\rm D}\right)}(m),\forall m$). For Type-3 models, $S_{\mathrm{D},l}$ is deterministic; hence, its CCDF is trivially upper-bounded.
As the following lemma rigorously shows, the mutual independence of $\left\{I_l\right\}_{l=0\cdots L-1}$ in Assumption \ref{assum:technical} guarantees strong delay domain sparsity, while Doppler sparsity is guaranteed in a more explicit manner by the CCDF upper bound.
\begin{lemma}
 \label{lem:DS_HS}
 With probability $1-e^{-\Omega\left(\min\left((2Q+1)p_{\rm D},Lp_{\rm d}\right)\right)}$, $\boldsymbol{\alpha}$ is $\left(s_{\rm d},s_{\rm D}\right)$-sparse under Assumption \ref{assum:technical}.
\end{lemma}
\begin{proof}
    The proof (which we first provided in \cite{sparse_spawc}) follows from applying Chernoff's bound to $S_{\rm d}$ evaluated at $s_{\rm d}=(1+\epsilon)Lp_{\rm d}$ and to the $\mathrm{B}\left(2Q+1,p_{\rm D}\right)$ distribution (upper-bounding $S_{\rm D}$ as per Assumption \ref{assum:technical}) evaluated at $s_{\rm D}=(1+\epsilon)(2Q+1)p_{\rm D}$ with an $\epsilon>0$ chosen as small as needed. 
\end{proof}

\subsection{Theoretical analysis of sparse recovery of DS-LTV channels with different waveforms}
The above established link to the paradigm of hierarchical sparsity provides us with theoretical tools for an analytical comparison of LTV channel estimation performance when different practical waveforms are used. The comparison is done assuming embedded pilots\footnote{As opposed to the joint channel estimation and data detection approach of \cite{Ranasinghe2025}, we only consider pilot-based channel estimation. While data-aided approaches can improve estimation accuracy, we believe that the relative performance comparison results we obtain assuming pure pilot-based channel estimation will continue to hold under the joint approach (though at a higher level of performance for each of the considered waveforms)}. For that sake, let $\mathbf{s}$ be the $N$-long vector of samples $s_n$ at the input of the channel, and define $\mathbf{x}$ as the vector of data symbols and embedded channel estimation pilots that modulate the waveform in use to produce the time-domain samples vector $\mathbf{s}$. For all the considered waveforms, we can write the dependence of $\mathbf{s}$ on $\mathbf{x}$ using a modulation matrix $\boldsymbol{\Phi_{\rm tx}}$
\begin{equation}
\label{eq:s_x_phi}
    \mathbf{s}=\boldsymbol{\Phi}_{\rm tx}\mathbf{x}\:.
\end{equation}
For single-carrier modulation (SCM), $\boldsymbol{\Phi}_{\rm tx} = \mathbf{T}_{\rm cp}$ while for AFDM, $\boldsymbol{\Phi}_{\rm tx} =\mathbf{T}_{\rm cpp} \pmb{\Lambda}_{c2} \mathbf{F}_N^{\rm H} \pmb{\Lambda}_{c1}$, with $\mathbf{T}_{\rm cp}$ and $\mathbf{T}_{\rm cpp}$ being the matrix for $(L-1)$-long CP insertion and chirp-periodic prefix (CPP) \cite{sparse_spawc} insertion, respectively, $\pmb{\Lambda}_{c} = \text{diag}(e^{-\imath 2\pi c n^2}, n = 0, \dots, N-1)$ and $\mathbf{F}_N$ being the $N$-order discrete Fourier transform (DFT) matrix. For the OFDM grid shown in Figure \ref{fig:OFDM_pilot_pattern} and composed of $N_{\rm ofdm,symb}$ symbols, each having $N_{\rm fft}$ sub-carriers and a cyclic prefix of length $N_{\rm cp}$, $\boldsymbol{\Phi}_{\rm tx}=\mathrm{blkdiag}\left(\mathbf{T}_{\rm cp} \mathbf{F}_{N_{\rm fft}}^{\rm H},\ldots,\mathbf{T}_{\rm cp} \mathbf{F}_{N_{\rm fft}}^{\rm H}\right)$.  For OTFS, $\boldsymbol{\Phi}_{\rm tx} = \mathbf{T}_{\rm cp}\left(\mathbf{F}_{M_{\rm otfs}}^{\rm H}\otimes\mathbf{I}_{N_{\rm otfs}}\right)$ where $M_{\rm otfs}$ and $N_{\rm otfs}$ are two integers satisfying $M_{\rm otfs}N_{\rm otfs}=N$.
%

At the receiver, let $\mathbf{r}$ be the $N$-long vector of the channel output samples $r_n$ defined in \eqref{eq:on_grid_discrete_td_IO}. Depending on the waveform employed, the receiver applies a demodulation matrix $\boldsymbol{\Phi}_{\rm rx}$ to $\mathbf{r}$ to produce the vector of transform domain samples $\mathbf{y}$. For SCM, $\boldsymbol{\Phi}_{\rm rx}=\mathbf{R}_{\rm cp}$ while for AFDM, $\boldsymbol{\Phi}_{\rm rx}=\pmb{\Lambda}_{c1}^{\rm H} \mathbf{F}_N \pmb{\Lambda}_{c2}^{\rm H}\mathbf{R}_{\rm cpp}$ with $\mathbf{R}_{\rm cp}=\mathbf{R}_{\rm cpp}=\left[\mathbf{0} , \mathbf{I}_{N}\right]$ being the matrices for $(L-1)$-long CP and CPP removal. For the OFDM grid of Figure \ref{fig:OFDM_pilot_pattern}, $\boldsymbol{\Phi}_{\rm rx}=\mathrm{blkdiag}\left(\mathbf{F}_{N_{\rm fft}} \mathbf{R}_{\rm cp},\ldots,\mathbf{F}_{N_{\rm fft}} \mathbf{R}_{\rm cp}\right)$. For OTFS, $\boldsymbol{\Phi}_{\rm rx}=\left(\mathbf{F}_{M_{\rm otfs}}\otimes\mathbf{I}_{N_{\rm otfs}}\right)\mathbf{R}_{\rm cp}$.
Let $\{{\rm p}_p\}_{p=1\cdots N_{\rm p}}$ be the $N_{\rm p}$ transform domain pilots, inserted at indexes $\left\{m_p\right\}_{p=1\cdots N_{\rm p}}$ within the vector $\mathbf{x}$, each surrounded by a waveform-dependent number of zero guard samples. Define $\mathcal{P}\subset\Iintv{0,N-1}$ as the set of indexes of the received samples associated with these pilots. 
For SCM, each time domain pilot instance needs a number of guard samples that is at least equal to the maximum delay shift (Fig. \ref{fig:SCM_pilot_pattern}). We thus set
\begin{equation}
    \label{eq:calP_scm}
    \mathcal{P}=\bigcup_{p=1}^{N_{\rm p}}\Iintv{m_p,m_p+L-1}.\quad\text{(SCM)}
\end{equation}
Note that the above pilot scheme with time domain embedded pilots can also work for waveforms other than SCM. This is the case, for instance, of interleave Frequency Division Multiplexing (IFDM) \cite{IFDM}, a waveform with a modulation matrix that is particularly relevant for good BER performance with orthogonal approximate message passing (OMAP) and memory approximate message mapping (MMAP) receivers. Therefore, SCM pilot overhead results, that we provide later on, also apply to such waveforms.
As for OFDM, let $m_{p_{\mathrm{t}}}\in\Iintv{0,N_{\rm ofdm,symb}-1}$ be the time domain position of the $p$-th pilot and $m_{p_{\mathrm{f}}}\in\Iintv{0,N_{\rm fft}-1}$ be its frequency domain position. Assuming that $N_{\rm fft}$ has been chosen small enough to render negligible the Doppler effect \emph{within} the individual symbols of the OFDM grid, then no guard samples are needed in the frequency domain, while the multiple CPs play the role of time domain guard intervals (Fig. \ref{fig:OFDM_pilot_pattern}). Then,
\begin{equation}
    \label{eq:calP_ofdm}
    \begin{multlined}
    \mathcal{P}=\left\{m_p=m_{\mathrm{t},p_{\mathrm{t}}}N_{\rm fft}+m_{\mathrm{f},,p_{\mathrm{f}}}\right\}_{\substack{p_{\rm t}=1\cdots N_{\rm p,t},\\p_{\rm f}=1\cdots N_{\rm p,f}}},\\
    N_{\rm p}=N_{\rm p,t}N_{\rm p,f},\text{(OFDM)}
    \end{multlined}
\end{equation}
While for AFDM with a chirp parameter $c_1$, it holds \cite{sparse_spawc} that an $l$ delay shift and a $q$ Doppler shift of a channel path produce a $q-2Nc_1 l$ DAFT domain shift. More precisely,
\begin{equation}
\begin{aligned}
   &y_k=\sum_{l = 0}^{L-1}\sum_{q=-Q}^{Q}\alpha_{l,q}I_{l,q} 
    e^{\imath2\pi(c_1l^2-\frac{ml}{N} + c_2(m^2 - k^2))}x_k +w_k\:,\\
    &m \triangleq (k - q + 2Nc_1 l)_N\:.
\end{aligned}
    \label{eq:y_output_integer}
\end{equation}
The samples related to the $p$-th pilot symbol thus occupy $2N \lvert c_1 \rvert (L-1)+2Q+1$ DAFT domain indexes. More precisely,  each pilot instance in the DAFT domain should be granted $Q$ guard samples preceding it and $2N \lvert c_1 \rvert (L-1)+Q$ guard samples following it (Fig. \ref{fig:AFDM_pilot_pattern}). We thus set
\begin{equation}
    \label{eq:calP_afdm}
    \mathcal{P}=\bigcup_{p=1}^{N_{\rm p}}\Iintv{m_p-Q,m_p+2N \lvert c_1 \rvert (L-1)+Q}.\quad\text{(AFDM)}
\end{equation}
Finally, for OTFS, the set $\mathcal{P}$ is defined as the vectorized form of the indexes of the cells marked in blue in Figure \ref{fig:pilot_pattern_otfs}. Note that, since OTFS is a full-diversity waveform for LTV channels, the dimensions of the set $\mathcal{P}$ are adjusted to the channel delay and Doppler spreads; its cardinality cannot thus be reduced below its $L-1+\min(4Q+1,N_{\rm otfs})\min(2L-1,M_{\rm otfs})$ value if pilot-data orthogonality is to be preserved.
\begin{figure}
    \centering
    \includegraphics[width=1\linewidth]{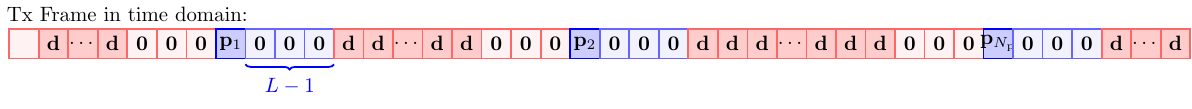}
    \caption{SCM frame composed of data samples and $N_{\rm p}$ pilot symbols, each of the latter surrounded by $2L-1$ guard samples.}
    \label{fig:SCM_pilot_pattern}
        \vspace{-2mm}
\end{figure}
\begin{figure}
    \centering
    \includegraphics[width=\linewidth]{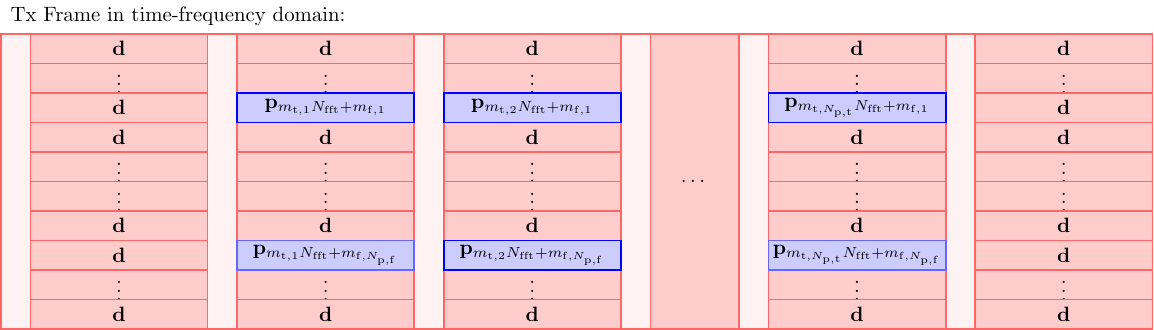}
    \caption{A frame of $N_{\rm ofdm,symb}$ $N_{\rm fft}$-long OFDM symbols, $N_{\rm p,t}$ of which having $N_{\rm p,f}$ pilot subcarriers
    }
    \label{fig:OFDM_pilot_pattern}
    \vspace{-2mm}
\end{figure}
\begin{figure}
    \centering
    \includegraphics[width=1\linewidth]{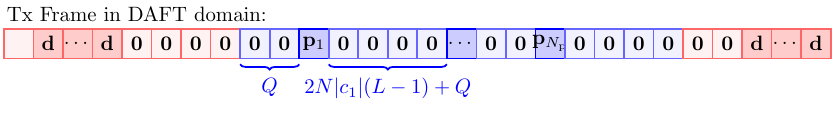}
    \caption{An AFDM symbol with $N_{\rm p}$ pilot symbols and their guard samples.}
    \label{fig:AFDM_pilot_pattern}
    \vspace{-2mm}
\end{figure}
\begin{figure}
  \centering
  \includegraphics[scale=.5]{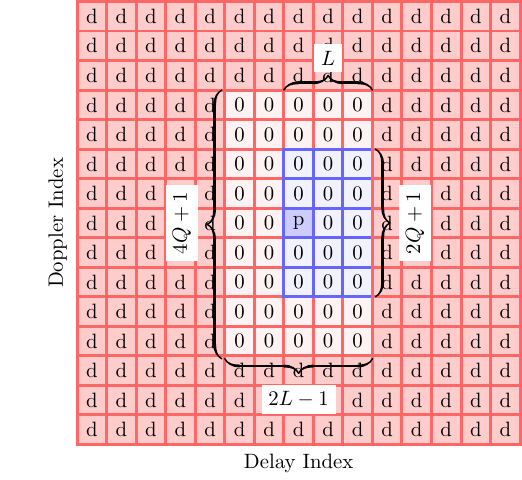}
  \caption{An OTFS symbol composed in the Zak domain of data samples (red), a pilot sample (blue) and guard samples (light blue and red)}
  \label{fig:pilot_pattern_otfs}
  \vspace{-2mm}
\end{figure} 

For any of the above waveforms, let $\mathbf{y}_{\rm p}\triangleq[y_k]_{k\in\mathcal{P}}$ be the vector of received pilot samples. Inserting \eqref{eq:ch_model} and \eqref{eq:s_x_phi} into \eqref{eq:on_grid_discrete_td_IO} gives the following signal model relating the hierarchically-sparse (per Lemma \ref{lem:DS_HS}) unknown channel vector $\boldsymbol{\alpha}$ to the measurement vector $\mathbf{y}_{\rm p}$ through a measurement matrix $\mathbf{M}_{\rm p}$:
\begin{equation}
     \label{eq:yp_ongrid}
     \mathbf{y}_{\rm p}=\underbrace{\mathbf{A}_{\cal P}\mathbf{M}}_{\triangleq\mathbf{M}_{\rm p}}\boldsymbol{\alpha}+\mathbf{w}_{\rm p}
 \end{equation}
where $[\mathbf{M}]_{l(2Q+1)+Q+q+1}=\boldsymbol{\Phi}_{\rm rx}\boldsymbol{\Delta}_q \boldsymbol{\Pi}^l \boldsymbol{\Phi}_{\rm tx} \mathbf{x}_{\rm p}$, $\mathbf{x}_{\rm p}$ is a vector of the same length as $\mathbf{x}$ with entries equal to ${\rm p}_1, \ldots, {\rm p}_{N_{\rm p}}$ at indexes $\left\{m_p\right\}_{p=1\cdots N_{\rm p}}$ and to zero elsewhere, and $\mathbf{w}_{\rm p}$ are the corresponding noise samples. Here, $\mathbf{A}_{\cal P}$ is the $|\mathcal{P}|\times N$ matrix that selects from the transform domain received vector the entries corresponding to $\mathcal{P}$.
$\boldsymbol{\Delta}_q=\mathrm{diag}(e^{\imath2\pi qn},n=0\cdots N-1)$, $\boldsymbol{\Pi}$ is the $N$-order permutation matrix, $\boldsymbol{\Phi}$ is the transform matrix specific to the waveform.
Hierarchical hard thresholding pursuit (HiHTP) \cite{hierarchical} has been suggested in the literature for solving  hierarchically-sparse recovery problems. When applied to Problem \eqref{eq:yp_ongrid} it gives Algorithm \ref{algo:HiHTP}.
\begin{algorithm}
    \caption{HiHTP for compressive sensing of DS-LTV channels}
    \label{algo:HiHTP}
    \begin{algorithmic}[1]
        \STATE \textbf{Input:} Measurement matrix $\mathbf{M}_{\rm p}$,  received pilot samples $\mathbf{y}_{\rm p}$, max. number of iterations $k_{\max}$, sparsity levels $s_{\rm d}$, $s_{\rm D}$
        \STATE $\hat{\boldsymbol{\alpha}}^{\left(0\right) }= 0$, $k=0$
        \STATE \textbf{repeat}
        \STATE $\Omega^{(k+1)}=L_{s_{\rm d},s_{\rm D}}\left(\boldsymbol{\alpha}^{\left(k\right)}+\mathbf{M}_{\rm p}^{\rm H}\left(\mathbf{y}_{\rm p}-\mathbf{M}_{\rm p}\boldsymbol{\alpha}^{\left(k\right)}\right)\right)$
        \STATE $\boldsymbol{\alpha}^{\left(k+1\right)}={\arg \min}\left\{\left\|\mathbf{y}_{\rm p}-\mathbf{M}_{\rm p}\mathbf{z}\right\|,\sup\left(\mathbf{z}\right)\subset\Omega^{(k+1)}\right\}$
        \STATE $k=k+1$
        \STATE \textbf{until} $k=k_{\max}$ or $\Omega^{(k+1)}=\Omega^{(k)}$ (whichever earlier)
        \STATE \textbf{Output:} $\left(s_{\rm d},s_{\rm D}\right)$-sparse $\hat{\boldsymbol{\alpha}}^{\left(k\right)}$.
        
    \end{algorithmic}
\end{algorithm}
HiHTP employs a \emph{hierarchically} sparse thresholding operator $L_{s_{\rm d},s_{\rm D}}$ at each iteration.
To compute $L_{s_{\rm d},s_{\rm D}}(\mathbf{x})$ for a vector $\mathbf{x}\in\mathbb{C}^{L(2Q+1)}$ a $s_{\rm D}$-sparse approximation is first applied to each of the $L$ blocks of $\mathbf{x}$ by keeping the largest $s_{\rm D}$ entries in each block while setting the remaining ones to zero. A $s_{\rm d}$-sparse approximation is then applied to the result by identifying the $s_{\rm d}$ blocks with the largest $l_2$-norm. HiHTP is known to converge, with a geometric convergence rate, provided that the associated HiRIP is small enough \cite[Theorem 1]{hierarchical}. This applies to Algorithm \ref{algo:HiHTP} (as established by Corollary \ref{cor:HiHTP_convergence} below) provided that the following technical assumption holds.
\begin{assumption}
    \label{assum:technical_diag}
    $\forall l_1\neq l_2$, random variables $\{I_{q}^{(l_1)}\}_{q=-Q\cdots Q}$ are independent of $\{I_{q}^{(l_2)}\}_{q=-Q\cdots Q}$.
\end{assumption}
Under the above assumption, the following two theorems characterize the pilot overhead needed for each of the studied waveforms in order to make the associated HiRIP small enough and hence to guarantee the convergence of Algorithm \ref{algo:HiHTP} and the consequent recovery of $\boldsymbol{\alpha}$.


\begin{theorem}[HiRIP for SCM and OFDM based measurements]
    \label{theo:HiRIP_SCM_OFDM}
    Under Assumption \ref{assum:technical} and for sufficiently large $L$, $Q$, sufficiently small $\delta_{\rm t}$, sufficiently small $\delta_{\rm f}$, and
    \begin{eqnarray}
        N_{\rm p,t}=\Omega(\frac{1}{\delta_{\rm t}^{2}}\log^2\frac{1}{\delta_{\rm t}}\log\frac{s_{\rm D}}{\delta_{\rm t}}s_{\rm D}\log(2Q+1))\:,\\
        N_{\rm p,f}=\Omega(\frac{1}{\delta_{\rm f}^{2}}\log^2\frac{1}{\delta_{\rm f}}\log\frac{s_{\rm d}}{\delta_{\rm f}}s_{\rm d}\log L)\:,
    \end{eqnarray}
    
    then the HiRIP constant $\delta_{s_{\rm d},s_{\rm D}}$ of the SCM measurement matrix $\mathbf{M}_{\rm p}$ associated with $\mathcal{P}\left(N_{\rm p,t}\right)$ as defined by \eqref{eq:calP_scm} satisfies $\delta_{s_{\rm d},s_{\rm D}}\leq\delta_{\rm t}$ with probability $1-e^{-\Omega \left(\log{L}\log{\frac{s_{\rm D}}{\delta_{\rm t}}}\right)}$ and the HiRIP constant $\delta_{s_{\rm d},s_{\rm D}}$ of the OFDM measurement matrix $\mathbf{M}_{\rm p}$ associated with $\mathcal{P}=\mathcal{P}\left(N_{\rm p,t},N_{\rm p,f}\right)$ as defined by \eqref{eq:calP_ofdm} satisfies $\delta_{s_{\rm d},s_{\rm D}}\leq\delta_{\rm t}+\delta_{\rm f}+\delta_{\rm t}\delta_{\rm f}$ with probability $1-e^{-\Omega\left(\min\left\{\log{L}\log{\frac{s_{\rm d}}{\delta_{\rm t}}},\log{(2Q+1)}\log{\frac{s_{\rm D}}{\delta_{\rm f}}}\right\}\right)}$.
\end{theorem}
\begin{proof}
    The proof of the theorem is given in Appendix \ref{app:proof_theo_scm_ofdm}.
\end{proof}
\begin{theorem}[HiRIP for AFDM based measurements]
\label{theo:HiRIP_AFDM}
Assume $\boldsymbol{\Phi}_{\rm tx} = \pmb{\Lambda}_{c2} \mathbf{F}_N \pmb{\Lambda}_{c1}=\boldsymbol{\Phi}_{\rm rx}^{\rm H}$, $\left|c_1\right|=\frac{P_{\rm afdm}}{2N}$ and let $P_{\rm afdm}$ be the \emph{smallest} integer satisfying $(L-1)P_{\rm afdm}+ 2Q+1\geq s_{\rm d}s_{\rm D}$ and $\mathcal{P}=\mathcal{P}\left(N_{\rm p}\right)$ be defined by \eqref{eq:calP_afdm}. Then, under Assumptions \ref{assum:technical} and \ref{assum:technical_diag} and for sufficiently large $L$, $Q$, sufficiently small $\delta$, and

\begin{equation}
    \begin{multlined}
        N_{\rm p}=
        \Omega(\frac{1}{\delta^{2}}\log^2\frac{1}{\delta}\log\frac{\log (LP_{\rm afdm})}{\delta}\log(LP_{\rm afdm})\log\frac{Q}{P_{\rm afdm}})\:,
    \end{multlined}
\end{equation}
the HiRIP constant $\delta_{s_{\rm d},s_{\rm D}}$ of matrix $\mathbf{M}_{\rm p}$ satisfies $\delta_{s_{\rm d},s_{\rm D}}\leq\delta$ with probability $1-e^{-\Omega \left(\log{\left(2\lceil\frac{Q}{P_{\rm afdm}}\rceil+1\right)}\log{\frac{\log(LP_{\rm afdm})}{\delta}}\right)}$. 
\end{theorem}
\begin{proof}
The proof of the theorem is given in Appendix \ref{app:proof_theo}.
\end{proof}
When $P_{\rm afdm}\triangleq2N|c_1|$ is set to $2Q+1$, AFDM achieves full diversity \cite{BemaniAFDM_TWC} and the measurements are non-compressive. The setting $P_{\rm afdm}=1$, on the other hand, is the most compressive. By choosing the value of $P_{\rm afdm}$ between these two extremes, as in the theorem, each pilot instance provides in its $(L-1)P_{\rm afdm}+2Q+1$-long guard interval a number of measurements that is close, with high probability, to the number $s_{\rm d}s_{\rm D}$ of unknowns. Furthermore, the theorem states that $N_{\rm p}$, the number of AFDM pilot instances required to estimate the sparsity support has only logarithmic growth with respect to \emph{both} delay and Doppler spreads. This property is to be contrasted with the first-degree polynomial dependence of $N_{\rm p}$ on $s_{\rm d}$ or $s_{\rm D}$ in the case of SCM and OFDM.
\begin{corollary}[Recovery guarantee for compressive sensing of DS-LTV channels]
\label{cor:HiHTP_convergence}
If $\mathbf{M}_{\rm p}$ satisfies the conditions of Theorems \ref{theo:HiRIP_SCM_OFDM} or \ref{theo:HiRIP_AFDM}, the sequence $\hat{\boldsymbol{\alpha}}^{(k)}$ defined by Algorithm \ref{algo:HiHTP} satisfies
$\|\hat{\boldsymbol{\alpha}}^{(k)}-\boldsymbol{\alpha}\|\leq \rho^k\|\boldsymbol{\alpha}^{(0)}-\boldsymbol{\alpha}\|+\tau\left\|\mathbf{w}_{\rm p}\right\|$
where $\rho<1$ and $\tau$ are constants defined in \cite[Theorem 1]{hierarchical}.
\end{corollary}
\begin{proof}
Thanks to Theorems \ref{theo:HiRIP_SCM_OFDM} and \ref{theo:HiRIP_AFDM}, matrix $\mathbf{M}_{\rm p}$ with large enough $L,Q,N_{\rm p}$ can be made to have a HiRIP constant that satisfies $\delta_{3s_{\rm d},2s_{\rm D}}<\frac{1}{\sqrt{3}}$. The conditions of \cite[Theorem 1]{hierarchical} are thus satisfied, and the corollary follows from that theorem.
\end{proof}
\begin{remark}
\label{rem:asym}
    To gain some insights into the above results, we examine them in an asymptotic regime defined with the help of an auxiliary variable $K$. The regime is characterized by $L\sim K$, $Q\sim K$, $s_{\rm d}\sim K^{\kappa_{\rm d}}$ and $s_{\rm D}\sim K^{\kappa_{\rm D}}$ for some $\kappa_{\rm d},\kappa_{\rm D}\in[0,1)$ where the notation $\sim$ stands for asymptotic equivalence. Note that, in this regime, the frame size satisfies $N\sim K^2$ due to the fact that $L,Q\sim K$ \cite{afdm_gc}. As for the fraction of $N$ dedicated to pilot and guard samples (including both prefix and data guard samples), its asymptotic behavior depends on the waveform. Due to Theorem \ref{theo:HiRIP_SCM_OFDM}, SCM overhead equals
    \begin{equation}
    \label{eq:scm_overhead}
        L-1+(2L-1)N_{\rm p,t}=\Omega\left((\log K)^2 K^{1+\kappa_{\rm D}}\right)\:\text{(SCM)}
    \end{equation}
    For OFDM, and due to the same theorem, the overhead equals
    \begin{equation}
    \label{eq:ofdm_overhead}
        \left(L-1+N_{\rm p,f}\right)N_{\rm p,t}=\Omega\left((\log K)^2 K^{1+\kappa_{\rm D}}\right)\:\text{(OFDM)}
    \end{equation}
    OFDM overhead is thus smaller than its SCM counterpart but has the same asymptotic growth. AFDM overhead satisfies
    \begin{equation}
    \label{eq:afdm_overhead}
    \begin{multlined}
        L-2+(N_{\rm p}+1)\left((L-1)P_{\rm afdm}+2Q+1\right)=\\
        \Omega\left(\log(\log K)(\log K)^2 K\right)
    \end{multlined}
    \text{(AFDM)}
    \end{equation}
    As for OTFS, its overhead satisfies
    \begin{equation}
    \label{eq:otfs_overhead}
        \begin{multlined}
            L-1+\min(4Q+1,N_{\rm otfs})\times\\
            \min(2L-1,M_{\rm otfs})=\Omega\left(K^2\right)
        \end{multlined}
        \text{(OTFS)}  
    \end{equation}
    
    Therefore, under delay-Doppler sparsity, AFDM overhead is asymptotically dominated by that of OFDM, SCM  and OTFS.
    \end{remark}
The next section is dedicated to relaxing the on-grid approximation to generalize the above results to more realistic propagation and channel prediction.

\section{Estimation and extrapolation of doubly sparse LTV channels with off-grid Doppler shifts}
\label{sec:off_grid}
While the on-grid approximation of Definition \ref{def:dD_sparsity} is useful for the conception of LTV channel estimation and sensing schemes and for the analysis of their performance, as we argued in Section \ref{sec:on_grid}, it lacks support for the finer Doppler resolution needed for channel prediction or the mitigation of channel aging. Indeed, most channel prediction paradigms \cite{haifan2022,prediction_cs2019,prediction_cs2024,channel_extension} involve, explicitly or implicitly, the estimation of Doppler frequency shifts to within an error margin smaller than the frequency resolution characteristic of the duration of the channel observation interval.
For that sake, we now present a second approximation for LTV channels that allows, in contrast to the first approximation in Definition \ref{def:dD_sparsity}, for fractional-valued Doppler frequency shifts.
\subsection{Off-grid doubly sparse linear time-varying channels}
We now allow for fractional-valued Doppler frequency shifts. In this case, only $\iota_p$ is zero, while $\kappa_p$ may take non-zero values sampled from a uniform distribution, i.e., $\kappa_p\sim\mathcal{U}([-\frac{1}{2},\frac{1}{2}]) $. Define $N_{\mathrm{D},l,q}$ as the number of ``sub-paths'', or ``rays'' using the terminology of clustered delay line (CDL) models \cite{3GPP}, which share the same integer part of their delays and Doppler shifts while only differing in the fractional part of their Doppler shifts:
\begin{equation}
    \label{NDq_offgrid}
    N_{\mathrm{D},l,q} \triangleq |\{ p\in\Iintv{1,N_{\rm p}} | l_p=l, q_p=q\}|\:.
\end{equation}
Depending on the scenario, it is in principle possible to model $N_{\mathrm{D},l,q}$ either as a fixed value or as a random variable, for example, a uniform random variable drawn over the range $\Iintv{1,N_{\rm D}}$. For simplicity, we opt for the first option, that is, $N_{\mathrm{D},l,q}=N_{\rm D},\forall l,q$ for some value $N_{\rm D}$ and we impose that the following channel power normalization should be satisfied.
\begin{equation}
    \label{eq:power_normalization_off_grid}
    \sum_{l=0}^{L-1}\sum_{q=-Q}^{Q}\sum_{i=1}^{N_{\mathrm{D}}}\mathbb{E}\left[\left|\alpha_{l,q,i}\right|^2I_{l,q}\right]=1\:.
\end{equation}

\begin{definition}[off-grid Delay-Doppler sparsity]
\label{def:offgrid_dD_sparsity}
The complex gain $h_{l,n}$ of the LTV channel $l$-th path varies with time as
 \begin{equation}
     \label{eq:off_ch_model}
     h_{l,n}=\sum_{q=-Q}^{Q}I_{l,q}\sum_{i=1}^{N_{\rm D}}\alpha_{l,q,i}e^{\imath 2\pi\frac{n(q+\kappa_{i})}{N}}, \quad l=0\cdots L-1
 \end{equation}
for some value $N_{\rm D}$, with \(I_{l,q}\) retaining the same description as in Definition \ref{def:dD_sparsity} and with $I_{l,q},\left\{\alpha_{l,q,i},\kappa_i\right\}_{i=1\cdots N_{\rm D}}$ being mutually independent. The complex gains satisfy $\alpha_{l,q,i}\sim\mathcal{CN}\left(0,\sigma_{\alpha}^2\right)$ with $\sigma_{\alpha}^2$ chosen so that \eqref{eq:power_normalization_off_grid} is respected.
\end{definition}
Figure \ref{fig:offgrid_examples} illustrates possible realizations of the above model. Note that the given examples can be seen as an off-grid relaxation of the three sparsity profiles from Figure \ref{fig:examples}.

\begin{remark}
Thanks to this relaxation, the channel model becomes compatible with CDL models. Indeed, a CDL channel is composed of rays, the delay and Doppler shifts of which have a random component around a mean value that is a function of the specific cluster to which these rays belong \cite{3GPP}. In our model, such rays are regrouped depending on the integer part of their normalized delay and Doppler frequency shifts, i.e., as a function of the nearest delay-Doppler grid point. The only difference is that, in our model, we allow for more randomness by introducing the random variables $I_{l,q}$ which determine which points of the delay-Doppler grid will be activated (while this activation is pre-determined in CDL models by the tables defining the mean delay and Doppler shifts characterizing each cluster of the model).
\end{remark}
\begin{remark}
    Maintaining the delay taps as integer values is not detrimental to the core message of this work, since the same approach we develop to deal with off-grid Doppler shifts can be extended to off-grid delays. Moreover, the effect of off-grid delays can still be made to fall under the current model by increasing the value of the model delay domain sparsity parameter $p_{\rm d}$ sufficiently to account for leakage due to the fractional part of the delay shifts.
\end{remark}

\begin{figure}
  \centering
  \begin{tabular}{ c @{\hspace{5pt}} c @{\hspace{5pt}} c}
  \includegraphics[width=.3\columnwidth] {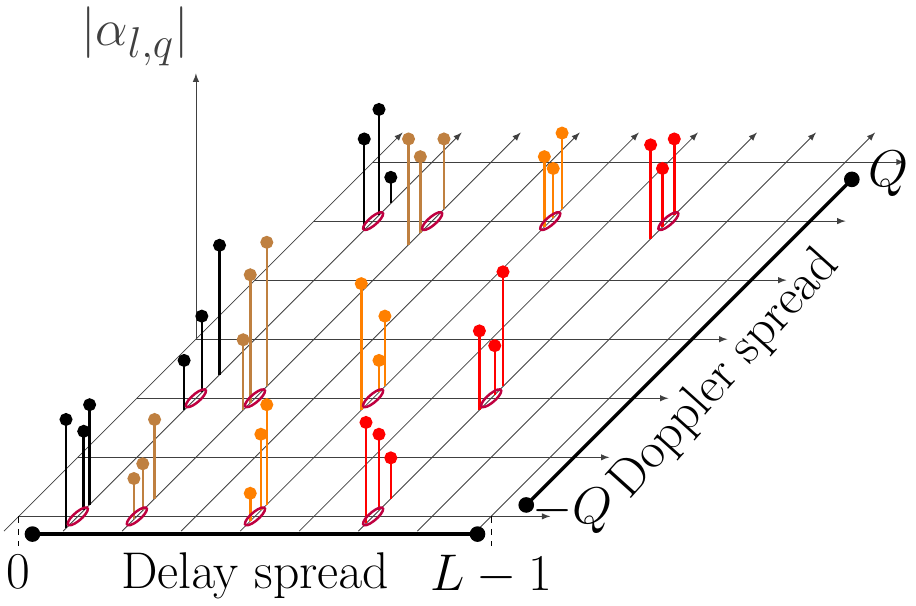} &
    \includegraphics[width=.3\columnwidth]{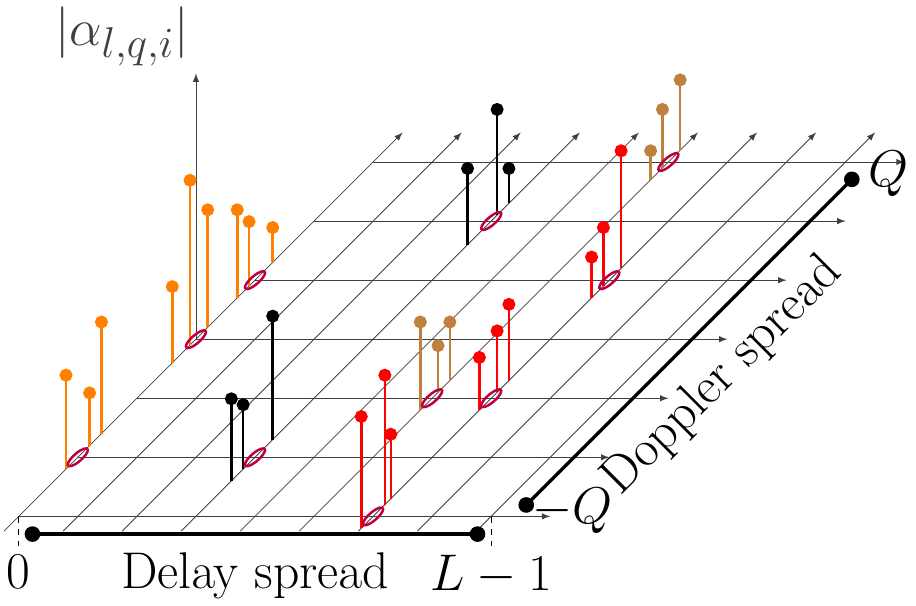} &
      \includegraphics[width=.3\columnwidth]{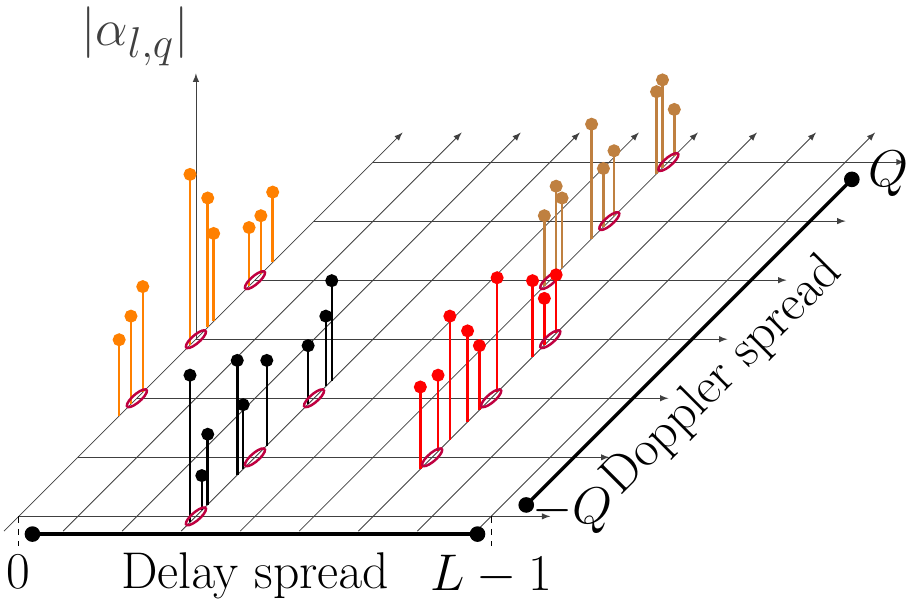} \\
    \small (a) &
      \small (b)&
      \small (c)
  \end{tabular}
  \caption{Three Instances of the delay-Doppler domain response of a DS-LTV channel satisfying Definition \ref{def:offgrid_dD_sparsity} and (a) off-grid Type-1, (b) off-grid Type-2, (c) off-grid Type-3 delay-Doppler sparsity}
  \label{fig:offgrid_examples}
  \vspace{-2mm}
\end{figure}


Estimating all the unknown parameters of the off-grid channel model in \eqref{eq:off_ch_model}, as in super-resolution approaches such as SBL and atomic-norm minimization \cite{sparse_sbl,kron_sbl,distributedCS_offgrid} or even just trying to identify them to the closest point of a finer delay-Doppler grid, as in refined-grid techniques \cite{Ganguly2019refinedgrid,Tian2016MicrowaveSC}, remains a challenging task. Indeed, the $N_{\rm D}$ Doppler frequency shifts $\{q+\kappa_i\}_{i=1\cdots N_{\rm D}}$ associated with a grid point $(l,q)$ differ only by their fractional parts, which renders the problem of estimating the corresponding complex coefficients $\{\alpha_{l,q,i}\}_{i=1\cdots N_{\rm D}}$ ill-conditioned. Moreover, the number of unknowns in the model, i.e., $N_{\rm p}$, could be prohibitively large. While grid refinement methods alleviate the dimensionality problem by constraining the sub-paths to a refined grid of a controlled size, they still suffer from ill-conditioning. For both these reasons, we propose a new model based on multiple ``elementary'' basis expansion models (BEM), each set with a bandwidth equal to the frequency resolution and shifted in frequency to be centered at one of the active grid points. Thanks to its optimality in terms of time-frequency localization \cite{fast_Slepian}, we opt for discrete prolate spheroidal sequences (DPSS) BEMs.

\subsection{Background: DPSS Basis Expansion Model}
 We begin with an overview of DPSS BEM modeling of a baseband signal $h_n$ that occupies a digital-frequency bandwidth $(-W,W)$ with $W\in\left(0,\frac{1}{2}\right)$. DPSS basis vectors $u_{b,n}^{(N,W)}$ ($b=1,\ldots,N$) are the eigenvectors of the prolate matrix \cite{BEM}:
\begin{equation}
    \label{eq:dpss_bem}
    \sum_{k=0}^{N-1}C_{k,n}^{(N,W)} u_{b,k}^{(N,W)}=\lambda_b^{(N,W)}u_{b,n}^{(N,W)}, n=0\cdots N-1
\end{equation}
where $C_{k,n}^{(N,W)}$ is the $(k-n)$-th entry of the prolate matrix:
\begin{equation}
    \label{eq:C}
    C_{k,n}^{(N,W)}= \frac{\sin{(2\pi W (k-n))}}{\pi (k-n)}\:.
\end{equation}
 The eigenvectors are normalized so that $\sum_{n=1}^{N}\left(u_{b,n}^{(N,W)}\right)^2=1$. The
eigenvalues are indexed in a decreasing order, $1 \geq \lambda_0^{(N,W)} \geq \ldots \geq \lambda_{N-1} ^{(N,W)}\geq 0$,  and display an energy-concentration property as $\lambda_b^{(N,W)}$ are clustered near $1$ for $b \leq 2WN$, and rapidly drop to zero for $b > 2WN$ \cite{DPSS}. 
In  what follows, we will employ multiple ``elementary'' DPSS BEMs, each of which is used to represent the channel signal component related to the $N_{\rm D}$ fractional Doppler shifts around one of the grid points $(l,q)$, i.e., 
$\sum_{i=1}^{N_{\rm D}} \alpha_{l,q,i} e^{\imath 2\pi \frac{n\kappa_{i}}{N}}$, and not the channel signal associated with the whole Doppler spread, i.e., $h_{l,n}$. 
Each of these BEMs is defined using \eqref{eq:dpss_bem} and \eqref{eq:C} with $W=\frac{1}{2N}$.
%
In the interest of simplicity, all parameters with the superscript $(N,W)$ in their notation, such as $u_{b,n}^{(N,W)}$, will be replaced by their simplified forms e.g., $u_{b,n}$.

\subsection{Off-grid modeling using multiple shifted elementary BEMs}
\label{sec:multi_shifted_BEM}
The complex gain $h_{l,n}$ of the $l$-th LTV channel path, as defined in \eqref{eq:off_ch_model}, is the sum of $2Q+1$ terms, each representing a signal occupying a digital-frequency band centered at $\frac{q}{N}$ and a bandwidth equal to $\frac{1}{N}$:
\begin{equation}
\label{eq:h_lqn}
    h_{l,q,n}\triangleq\sum_{i=1}^{N_{\rm D}}\alpha_{l,q,i}e^{\imath 2\pi\frac{n(q+\kappa_{i})}{N}}.
\end{equation}
The baseband version of $h_{l,q,n}$, defined as $e^{-\imath 2\pi\frac{nq}{N}}h_{l,q,n}$, can be modeled (see Figure \ref{fig:Q_bem_representation}) using a DPSS BEM of an order equal to $Q_{\rm BEM}$ (chosen large enough depending on the required modeling precision) and a bandwidth $(-W,W)$ with $W = \frac{1}{2N}$ by defining
\begin{equation}
    \label{eq:h_lqn_bem}
    \begin{multlined}
     h_{l,q,n}^{\rm BEM}\triangleq e^{\imath 2\pi\frac{nq}{N}}
     \sum_{b=1}^{Q_{\rm BEM}} \beta_{l,q,b} u_{b,n},\\
     \text{($q$-th shifted elementary BEM)}
     \end{multlined}
\end{equation}
 \begin{figure}
  \centering
  \begin{tabular}{ c  }
  \includegraphics[width=0.9\columnwidth]{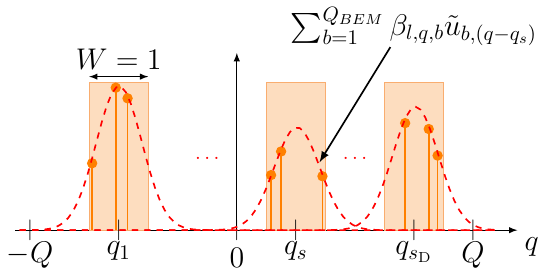} \\
  \small (a) stems: off-grid Doppler shifts, $\tilde{u}_{b,k}$: the DFT of $u_{b,n}$\\
    \includegraphics[width=1.0\columnwidth]{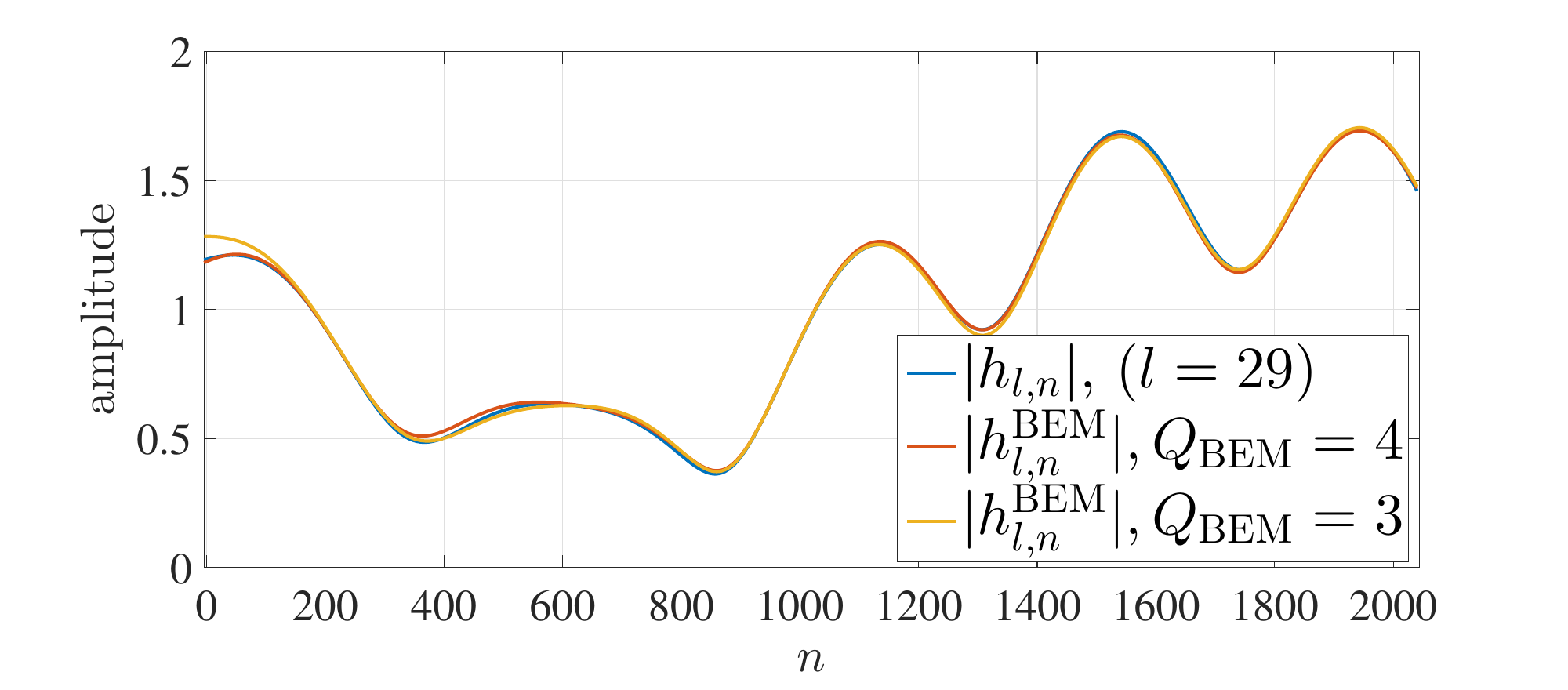} \\
      \small (b) $N=2048$, $W=\frac{1}{2N}$, $Q=7$, $S_{\rm D}=3$
  \end{tabular}
  \caption{(a)  The use of $S_{\rm D}$ frequency shifted copies of an elementary DPSS BEM $\left\{u_{b,n}\right\}_{b=1\cdots Q_{\rm BEM}}$ to capture the leakage due to $S_{\rm D}$ clusters of off-grid Doppler shifts $\left\{q_s+\kappa_i\right\}_{\substack{s=1\cdots S_{\rm D},\\i=1\cdots N_{\rm D}}}$ of an LTV channel tap $h_{l,n}$ (b) Precision of  this representation}
  \label{fig:Q_bem_representation}
  \vspace{-2mm}
\end{figure}

%

where the complex exponential $e^{\imath 2\pi\frac{nq}{N}}$ is used to shift, in frequency, the BEM term back to the digital frequency $\frac{q}{N}$. In vector form
\begin{align}
    \boldsymbol{\beta}_{l,q}&=\mathbf{U}_{Q_{\rm BEM}}^{\rm H}\mathbf{E}_{\frac{q}{N}}^{\rm H}\mathbf{h}_{l,q}\label{eq:BEM1}\\
    \mathbf{h}_{l,q}^{\rm BEM}&= \mathbf{E}_{\frac{q}{N}}\mathbf{P}^{\rm BEM} \mathbf{E}_{\frac{q}{N}}^{\rm H}\mathbf{h}_{l,q}\label{eq:BEM2}
\end{align}
with $\mathbf{U}_{Q_{\rm BEM}}\triangleq\left[\mathbf{u}_{1} \ldots \mathbf{u}_{Q_{\rm BEM}} \right]$ and $\mathbf{P}^{\rm BEM}=\mathbf{U}_{Q_{\rm BEM}}\mathbf{U}_{Q_{\rm BEM}}^{\rm H}$ being the orthogonal projection matrix associated with the DPSS BEM and $\mathbf{E}_{f}\triangleq\mathrm{diag}\left( e^{\imath 2\pi f 0} \ldots e^{\imath 2\pi f(N-1)}\right)$.
Inserting the shifted elementary BEMs, defined in \eqref{eq:h_lqn_bem} for each Doppler grid point, into \eqref{eq:off_ch_model} gives
\begin{equation}
\label{eq:ch_model_off_grid}
    \begin{multlined}
         h_{l,n}^{\rm BEM}=\sum_{q=-Q}^{Q} I_{l,q} e^{\imath 2\pi\frac{n q}{N}} \sum_{b=1}^{Q_{\rm BEM}} \beta_{l,q,b} u_{b,n},\\
    =\sum_{q=-Q}^{Q} I_{l,q}h_{l,q,n}^{\rm BEM},\:l=0\cdots L-1\\
    \text{(``multiple shifted elementary BEMs'' model)}
    \end{multlined}
\end{equation}

The precision of the above in relation to the order of its elementary BEMs is provided by the following theorem.
\begin{theorem}
    \label{theo:bem_accuracy}
    For any $\epsilon>0$, if we set $Q_{\rm BEM}\geq C\log\frac{1}{\epsilon}$ for some constant $C$ then the error of the model in \eqref{eq:ch_model_off_grid} satisfies $\sum_{l=0}^{L-1}\mathbb{E}[|h_{l,n}-h_{l,n}^{\rm BEM}|^2]<\epsilon$ for a sufficiently large $N$.
\end{theorem}
\begin{proof}
    The proof of Theorem \ref{theo:bem_accuracy} is given in Appendix \ref{app:bem_accuracy}.
\end{proof}
Theorem \ref{theo:bem_accuracy} states that the number of DPSS basis functions needed to represent the channel component associated with a single delay-Doppler grid point grows only logarithmically with the inverse of the target precision. The following figure (as well as the numerical results of Section \ref{sec:simulations}) shows that as few as 4 DPSS basis functions ($Q_{\rm BEM}=4$) are sufficient to achieve good-enough precision. Note that this value is unrelated to the Doppler spread value $2Q+1$: a larger Doppler spread does not imply the need for a larger value of $Q_{\rm BEM}$.

\subsection{DS-LTV off-grid LMMSE channel estimation}
 After relaxing the integer-Doppler assumption, thanks to the ``multiple shifted elementary BEMs'' model, we now undertake another step towards the goal of practical channel estimation and prediction for sparse LTV channels. Indeed, while the compressed-sensing paradigm provided us with theoretical tools that helped in the analytical comparison of the performance of different waveforms, it is preferable to perform channel estimation using  LMMSE due to its computational-complexity advantage.
Let $\boldsymbol{\beta}$ be the vectorized form of the BEM coefficients $\{\beta_{l,q,b}\}_{\substack{b=1\cdots Q_{\rm BEM}\\l,q,I_{l,q}\neq0}}$ associated with the active delay-Doppler grid points. Note that the knowledge of the delay-Doppler sparsity support, i.e., of $\left\{I_{l,q}\right\}_{l=0\cdots L-1,q=-Q\cdots Q}$, at the receiver side is now assumed.
Define $\mathbf{B}$ as the following block-diagonal matrix
\begin{equation}
\label{eq:B}
\mathbf{B}=
\mathrm{blkdiag}\left(\widetilde{\mathbf{B}},\ldots,\widetilde{\mathbf{B}}\right)
\end{equation}
where $\widetilde{\mathbf{B}}\triangleq[\widetilde{\mathbf{B}}_{-Q}\ldots\widetilde{\mathbf{B}}_{Q}]$ is a $N \times (2Q+1)Q_{\rm BEM}$ matrix, $[\widetilde{\mathbf{B}}_{q}]_{n,b}=\frac{1}{\sqrt{N}} \tilde{U}_{b,(n-q)}$ and $\tilde{U}_{b,n}=\frac{1}{\sqrt{N}}\sum_{k = 0}^{N-1} u_{b,k} e^{-\imath 2\pi \frac{nk}{N}}$ is the DFT of the DPSS basis vector $u_{b,n}$. Let $\mathbf{A}_{\beta}$ be the matrix that places the blocks of $\boldsymbol{\beta}$, each of size $Q_{\rm BEM}$, within the positions corresponding to $I_{l,q} \neq 0$ in a larger vector of length $LQ_{\rm BEM}(2Q+1)$ resulting in a $LQ_{\rm BEM}(2Q+1)$-long vector $\mathbf{A}_{\beta}\boldsymbol{\beta}$ that is block sparse. Now, define
\begin{equation}
    \label{eq:alpha}
    \boldsymbol{\alpha}\triangleq \mathbf{B} \mathbf{A}_{\beta} \boldsymbol{\beta}\:.
\end{equation}
Inserting \eqref{eq:alpha} into \eqref{eq:yp_ongrid}, we can write the received pilot samples vector $\mathbf{y}_{\rm p}$ as
\begin{equation}
    \label{eq:yp}
     \mathbf{y}_{\rm p}=\underbrace{\mathbf{A}_{\cal P}\mathbf{M}\mathbf{B}\mathbf{A}_{\beta}}_{\triangleq\mathbf{M}_{\rm p}}\boldsymbol{\beta}+\mathbf{w}_{\rm p}\:.
\end{equation}
where $\mathbf{M}$ and $\mathbf{A}_{\cal P}$ are defined in \eqref{eq:yp_ongrid}.
The MMSE estimate, $\hat{\boldsymbol{\beta}}$, of $\boldsymbol{\beta}$ based on $\mathbf{y}_{\rm p}$ is given by \cite{kay2013fundamentals} as
 \begin{equation}
     \label{eq:lmmse_beta}
\hat{\boldsymbol{\beta}} = \mathbf{B} \left( \sigma_{\alpha}^2 \mathbf{M}_{\rm p}^H \mathbf{M}_{\rm p} + \sigma_w^2 \mathbf{I} \right)^{-1} \mathbf{M}_{\rm p}^H \mathbf{y}_{\rm p}\:.
 \end{equation}
Finally, define
 \begin{equation}
    \label{eq:hat_h_beta_def}
    \hat{h}_{l,n}^{\rm BEM}=\sum_{q=-Q}^{Q} I_{l,q} e^{\imath 2\pi\frac{n q}{N}} \sum_{b=1}^{Q_{\rm BEM}} \hat{\beta}_{l,q,b} u_{b,n}, \quad n=0,\ldots, N-1\:.
 \end{equation}
 as the resulting MMSE estimate of $h_{l,n}^{\rm BEM}$. In vector form
 \begin{equation}
    \label{eq:hat_vec_h_beta_def}
    \hat{\mathbf{h}}_{l}^{\rm BEM}=\sum_{q=-Q}^{Q} I_{l,q}\underbrace{\mathbf{E}_{\frac{q}{N}}\mathbf{U}_{Q_{\rm BEM}}\hat{\boldsymbol{\beta}}}_{\triangleq\hat{\mathbf{h}}_{l,q}^{\rm BEM}}\:.
 \end{equation}

\begin{remark}[Estimation computational complexity]
\label{rem:estimation_complexity}
    The first dominant operation, complexity-wise, in the computation of $\hat{\mathbf{h}}_{l}^{\rm BEM}$ is the matrix inversion in \eqref{eq:lmmse_beta}. The complexity order of this inversion is cubic in the size of the vector $\boldsymbol{\beta}$, i.e., $\left(s_{\rm d}s_{\rm D}Q_{\rm BEB}\right)^3\sim K^{3(\kappa_{\rm d}+\kappa_{\rm D})}$. Here, $s_{\rm d}$ and $s_{\rm D}$ are the delay domain and Doppler domain sparsity levels defined by \eqref{eq:sd} and \eqref{eq:sD}, respectively, while $\kappa_{\rm d}$ and $\kappa_{\rm D}$ are the parameters of their asymptotic order, as defined by Remark \ref{rem:asym}. Note that the above value is much smaller than $\left(L(2Q+1)Q_{\rm BEM}\right)^3$, the computational complexity of LMMSE channel estimation under no sparsity. The second dominant operation in the computation of $\hat{\mathbf{h}}_{l}^{\rm BEM}$ is matrix $\mathbf{B}$-vector multiplication, an operation that costs $Ns_{\rm d}s_{\rm D}Q_{\rm BEB}\sim K^{2+\kappa_{\rm d}+\kappa_{\rm D}}$. Whether the first or the second operation is the dominant term depends on whether $\kappa_{\rm d}+\kappa_{\rm D}$ is smaller or larger than 1.
\end{remark}

The mean squared error (MSE) conditioned on a given $I_{l,q}$ realization is then written as
\begin{align}
    \label{eq:MSE_BEM_hat}
    &\sum_{l=0}^{L-1}\frac{1}{N}\mathbb{E}\left[\left\|\mathbf{h}_{l}^{\rm BEM}-\hat{\mathbf{h}}_{l}^{\rm BEM}\right\|^2\right]\nonumber\\
    &\leq\nonumber\frac{1}{N}\sum_{l=0}^{L-1}\sum_{q=-Q}^{Q}I_{l,q}\mathbb{E}\left[\left\|\mathbf{h}_{l.q}^{\rm BEM}-\hat{\mathbf{h}}_{l,q}^{\rm BEM}\right\|^2\right]\\
    &=\frac{1}{NN_{\rm D}\sigma_\alpha^2}\mathbb{E}\left[\left\|\boldsymbol{\beta}-\hat{\boldsymbol{\beta}}\right\|^2\right]
\end{align}
where the inequality is due to the triangle inequality and where the equality is due to the definition of $\hat{\mathbf{h}}_{l,q}^{\rm BEM}$ in \eqref{eq:hat_vec_h_beta_def}, to the fact that $\sum_{l=0}^{L-1}\sum_{q=-Q}^{Q}I_{l,q}=\frac{1}{N_{\rm D}\sigma_\alpha^2}$ per \eqref{eq:power_normalization_off_grid} and to the fact that $\mathbf{E}_{\frac{q}{N}}^{\rm H}\mathbf{E}_{\frac{q}{N}}=\mathbf{I}_N$ and $\mathbf{U}_{Q_{\rm BEM}}^{\rm H}\mathbf{U}_{Q_{\rm BEM}}=\mathbf{I}_{Q_{\rm BEM}}$.
\begin{assumption}
    \label{assum:Np}
    The number $N_{\rm p}$ of pilots is sufficiently large for the MSE associated with estimating $\boldsymbol{\beta}$ (and hence $\mathbf{h}_{l}^{\rm BEM}$ per \eqref{eq:MSE_BEM_hat}) to converge to zero as $\sigma_w^2\to0$.
\end{assumption}
\begin{remark}[Impact of the sensing waveform]\label{rem:waveform}
    The value of $N_{\rm p}$ needed for Assumption \ref{assum:Np} to hold depends on the underlying waveform. The numerical results given in Section \ref{sec:simulations} show that the relative advantage of AFDM over OFDM and SCM, which has been analytically established by Theorems \ref{theo:HiRIP_SCM_OFDM} and \ref{theo:HiRIP_AFDM} in the on-grid case, is still valid in the off-grid case in the form of better MSE performance at comparable pilot overhead levels.
\end{remark}
Assumption \ref{assum:Np} is about the MSE of estimating the BEM representation of the channel. The following corollary to Theorem \ref{theo:bem_accuracy} provides an analysis of the estimation MSE when computed with respect to the actual channel defined by \eqref{def:offgrid_dD_sparsity} instead of its BEM representation.
\begin{corollary}
    \label{cor:MSE}
    Under Assumption \ref{assum:Np} and Definition \ref{def:offgrid_dD_sparsity}, provided that $N$ is large enough and $Q_{\rm BEM}\geq C\log\frac{1}{\epsilon}$ for $\epsilon>0$ and some constant $C$, then $\lim_{\sigma_w^2\to0}\mathbb{E}[\sum_{l=0}^{L-1}\frac{1}{N}\|\mathbf{h}_{l}-\hat{\mathbf{h}}_{l}^{\rm BEM}\|^2]\leq\epsilon$.
\end{corollary}
\begin{proof}
    Apply the triangle inequality to $\mathbf{h}_{l}-\hat{\mathbf{h}}_{l}^{\rm BEM}=\left(\mathbf{h}_{l}-\mathbf{h}_{l}^{\rm BEM}\right)+\left(\mathbf{h}_{l}^{\rm BEM}-\hat{\mathbf{h}}_{l}^{\rm BEM}\right)$ followed by applying Assumption \ref{assum:Np} to the first term and Theorem \ref{theo:bem_accuracy} to the second.
\end{proof}
\subsection{Application to channel extrapolation and prediction}



The literature on the Slepian basis \cite{channel_extension} establishes a ``natural'' way to extend the finite sequences $u_{b,n}$ from the smaller interval $\Iintv{0,N-1}$ to the larger one $\Iintv{-N_{\text{ext}}, N + N_{\text{ext}}}$, where $N_{\text{ext}}$ denotes the additional channel samples to be extrapolated. This is achieved by letting the index $n$ in \eqref{eq:dpss_bem} be defined over $\mathbb{Z}$ instead of being confined to $\Iintv{0,N-1}$ leading to
\begin{equation}
    \label{eq:extension}
    u_{b,n}^{\rm ext}\triangleq\frac{1}{\lambda_{b}^{(N,W)}}\sum_{k=0}^{N-1} C_{k,n}^{(N,W)}u_{b,k}, n\in \mathbb{Z}\:.
\end{equation}
Signal $(u_{b,n}^{\rm ext})_{n\in\mathbb{Z}}$ (``$\mathrm{ext}$'' stands for ``extrapolation'') has a discrete-time Fourier transform (DTFT) that is zero outside $(-W,W)$ and is the signal that possesses the least energy outside the time interval $\Iintv{0,N-1}$ from among all the discrete-time signals band-limited to $(-W,W)$ \cite{channel_extension}. 




Once we have estimated the multiple-BEM representation of the off-grid DS-LTV channel on the interval $\Iintv{0,N-1}$ as in \eqref{eq:hat_h_beta_def} and once we have calculated the infinite-time version of the DPSS basis function as in \eqref{eq:extension}, the channel can thus be extrapolated for $n\in\mathbb{Z}$ as follows
\begin{equation}
\label{eq:dpss_ch_predictor}
    h_{l,n}^{\rm ext}\triangleq\sum_{q=-Q}^{Q} I_{l,q} \underbrace{e^{\imath 2\pi\frac{n q}{N}} \sum_{b=1}^{Q_{\rm BEM}} \hat{\beta}_{l,q,b} u_{b,n}^{\rm ext}}_{\triangleq h_{l,q,n}^{\rm ext}},\:l=0\cdots L-1\:.
\end{equation}
Here, we defined
\begin{equation}
    \label{eq:h_lqn_ext}
    h_{l,q,n}^{\rm ext}\triangleq
    e^{\imath 2\pi\frac{n q}{N}} \sum_{b=1}^{Q_{\rm BEM}}\hat{\beta}_{l,q,b} u_{b,n}^{\rm ext}\:,n\in\mathbb{Z}\:.
\end{equation}
Defining $\mathbf{u}_n^{\rm ext}\triangleq\left[u_{1,n}^{\rm ext}\quad\cdots\quad u_{Q_{\rm BEM},n}^{\rm ext}\right]^{\rm T}$, \eqref{eq:hat_vec_h_beta_def} gives
\begin{equation}
    \label{eq:h_lqn_ext_u}
    h_{l,q,n}^{\rm ext}=e^{\imath 2\pi\frac{n q}{N}}\left(\mathbf{u}_n^{\rm ext}\right)^{\rm T}\mathbf{U}_{Q_{\rm BEM}}^{\rm H}\mathbf{E}_{\frac{q}{N}}^{\rm H}\hat{\mathbf{h}}_{l,q}^{\rm BEM}\:.
\end{equation}
\vspace{-4.5mm}
\begin{remark}[Prediction computational complexity]
    The dominant cost, in terms of computational complexity, when computing $h_{l,n}^{\rm ext}$ using \eqref{eq:dpss_ch_predictor} is the number of complex multiplications needed for the computation of the extrapolated DPSS samples $\{u_{b,n}^{\rm ext}\}_{b=1\cdots Q_{\rm BEM}}$ using \eqref{eq:extension}, which is of the order of $Q_{\rm BEM}N\sim K^2$. Note that this complexity order is smaller than that of channel estimation, as established by Remark \ref{rem:estimation_complexity}.
\end{remark}

It is known \cite{fast_Slepian} that the $b$-th infinite-length DPSS has $1-\lambda_b^{(N,W)}$ (respectively $\lambda_b^{(N,W)}$) of its energy outside (respectively inside) the interval $\Iintv{0,N-1}$ as illustrated in Figure \ref{fig:dpss_ext_content}.
\begin{figure}
    \centering
    \includegraphics[width=1.0\columnwidth]{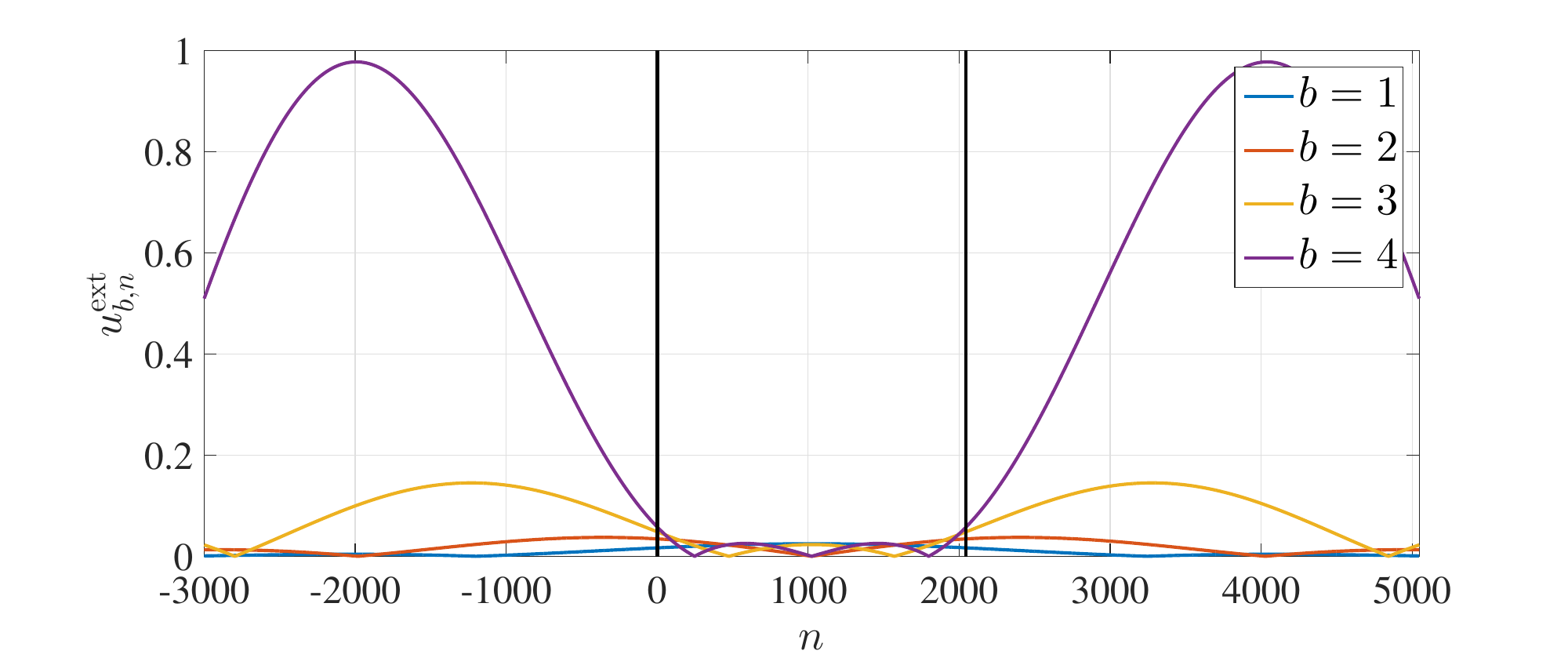}
    \caption{Extrapolated versions of the first four DPSS ($N=2048$, $W=\frac{1}{2N}$)}
    \label{fig:dpss_ext_content}
    \vspace{-2mm}
\end{figure}
Therefore, the terms of the sum that contribute the most to the value of $h_{l,q,n}^{\rm ext}=e^{\imath 2\pi\frac{n q}{N}} \sum_{b=1}^{Q_{\rm BEM}} \beta_{l,q,b} u_{b,n}^{\rm ext}$ for $n>N-1$ are the ones with large enough $u_{b,n}^{\rm ext}$, i.e., with small enough $\lambda_b^{(N,W)}$. While this constrains $Q_{\rm BEM}$ to be sufficiently large, it also constrains the SNR to be high enough to guarantee a precise estimation of the coefficients $\beta_{l,q,b}$ associated with basis functions that have little contribution to the signal received within the observation interval $\Iintv{0,N-1}$. The issue of setting the value of $Q_{\rm BEM}$ is further discussed in Section \ref{sec:simulations}.
Theorem \ref{theo:prediction} establishes the optimality of the predictor $h_{l,n}^{\rm ext}$ by relating it to reduced-rank MMSE.
\begin{theorem}
    \label{theo:prediction}
    Under the assumptions of Definition \ref{def:offgrid_dD_sparsity} and Assumption \ref{assum:Np} and for any $n>N-1$, the predictor $h_{l,n}^{\rm ext}$ defined by \eqref{eq:dpss_ch_predictor} converges, in the squared-mean sense in the limit of a vanishing noise variance, to a reduced-rank MMSE estimator of $h_{l,n}$ with a reduced rank equal to $Q_{\rm BEM}$.
\end{theorem}
\begin{proof}
The proof of Theorem \ref{theo:prediction} is given in Appendix \ref{app:prediction_proof}.    
\end{proof}
\vspace{-4.5mm}
\begin{remark}[Relation to other DPSS extrapolation prediction methods]\label{rem:DPSS_prediction}
    The DPSS-based multi-band single-BEM modeling in \cite{channel_extension} uses a single BEM of order $\Check{Q}_{\rm BEM}$ to represent the multi-band signal $h_{l,n}=\sum_{q=-Q}^{Q}I_{l,q}\sum_{i=1}^{N_{\rm D}}\alpha_{l,q,i}e^{\imath 2\pi\frac{n(q+\kappa_{i})}{N}}$ of each channel delay tap to give
    \begin{equation}
        \label{eq:check_h_beta_def}
        \Check{h}_{l,n}^{\rm BEM}=\sum_{b=1}^{\Check{Q}_{\rm BEM}} \Check{\beta}_{l,q,b} \Check{u}_{b,n}, \quad n=0,\ldots, N-1\:.
    \end{equation}    
    where $\Check{u}_{b,k}^{(N,W)}$ are the basis vectors of the single multi-band BEM defined for $b=1,\ldots, N$ as
    \begin{equation}
    \label{eq:dpss_onebem}
    \sum_{k=0}^{N-1}\Check{C}_{k,n}^{(N,W)} \Check{u}_{b,k}=\Check{\lambda}_b^{(N,W)}\Check{u}_{b,n}, n=0\cdots N-1\:,
\end{equation}
\begin{equation}
    \label{eq:C_onebem}
    \Check{C}_{k,n}^{(N,W)}\triangleq\sum_{\substack{q=-Q\\ I_{l,q}=1}}^{Q}e^{\imath 2 \pi\frac{q(k-n)}{N}}\frac{\sin{(2\pi W (k-n))}}{\pi (k-n)}\:.
\end{equation}
The associated channel predictor is 
\begin{equation}
\label{eq:single_dpss_ch_predictor}
    \Check{h}_{l,n}^{\rm ext}\triangleq \sum_{b=1}^{\Check{Q}_{\rm BEM}} \Check{\beta}_{l,q,b} \Check{u}_{b,n}^{\rm ext},\:l=0\cdots L-1\:,\:n\in\mathbb{Z}\:,
\end{equation}
where the associated extrapolated basis vector is given by
\begin{equation}
    \label{eq:single_dpss_extension}
    \Check{u}_{b,n}^{\rm ext}\triangleq\:\frac{1}{\Check{\lambda_{b}}^{(N,W)}}\sum_{k=0}^{N-1} \Check{C}_{k,n}^{(N,W)}\Check{u}_{b,k}, n\in \mathbb{Z}\:.
\end{equation}
In terms of computational complexity, the extrapolation of the multi-band single-BEM approach as done in \eqref{eq:single_dpss_ch_predictor} is equivalent to the extrapolation step of the ``multiple shifted elementary BEMs'' approach in \eqref{eq:dpss_ch_predictor}. The main issue with the multi-band single-BEM approach is the size of the codebook that needs to be computed or stored at the network device performing channel estimation. As can be seen from \eqref{eq:yp}, the codebook size in our approach is the number of columns of the matrix $\mathbf{MB}$ which equals $Q_{\rm BEM}L(2Q+1)$. In the case of the multi-band single-BEM approach, the codebook size is $L\sum_{k=1}^{2Q+1}{2Q+1\choose k}\Check{Q}_{\rm BEM}(k)\gg Q_{\rm BEM}L(2Q+1)$ as every different combination of $k$ active Doppler grid points would result in a different multi-band prolate matrix $\Check{C}_{k,n}^{(N,W)}$ \eqref{eq:C_onebem} and thus in a different DPSS basis \eqref{eq:dpss_onebem}.
\end{remark}

\section{Numerical results}
\label{sec:simulations}
We first present the results pertaining to the estimation of DS-LTV channels with the on-grid approximation of Definition \ref{def:dD_sparsity}. AFDM sparse recovery performance is compared to that of OFDM and OTFS. Note that OFDM pilot overhead results can be indicative of the pilot overhead performance of SCM and other waveforms using time domain embedded pilots such as IFDM.
We used 100 realizations of channels having a Type-1 delay-Doppler sparsity with $p_{\rm d}=0.2$, $p_{\rm D}\in\{0.2,0.4\}$ and $N=4096, L=30, Q=7$ (corresponding to a 30 MHz transmission at a 70 GHz carrier frequency, a maximum target moving speed of 396 km/h and a maximum target range of 300 meters, or equivalently to 10 MHz transmission at a 30 GHz carrier frequency, a maximum speed of 57 km/h and a maximum target range of 900 meters). For both OFDM and AFDM, sparse recovery of $\boldsymbol{\alpha}$ is done using HiHTP (Algorithm \ref{algo:HiHTP}) based on the pilot pattern described in Figures \ref{fig:OFDM_pilot_pattern} and \ref{fig:AFDM_pilot_pattern}, respectively. For OTFS, since sensing is done without compression, non-compressive estimation algorithms can be used based on the pilot pattern of Fig. \ref{fig:pilot_pattern_otfs} \cite{BemaniAFDM_TWC}. For each waveform, the pilot overhead was set in such a way that the mean squared error ${\rm MSE}\triangleq\mathbb{E}[\left\|\hat{\boldsymbol{\alpha}}-\boldsymbol{\alpha}\right\|^2]$ is approximately $10^{-4}$ at $\mathrm{SNR}=20$ dB. Fig. \ref{fig:simulations1_2} shows an advantage of AFDM in terms of pilot overhead, i.e., the number of samples in each frame needed as pilots and guards to achieve the target MSE performance. This performance validates the theoretical result, given in Theorems \ref{theo:HiRIP_SCM_OFDM} and \ref{theo:HiRIP_AFDM}, implying that the overhead needed to achieve the same HiRIP level is smaller for AFDM compared to SCM and OFDM.
 \begin{figure}
  \centering
  \begin{tabular}{ c  }
  \includegraphics[width=0.9\columnwidth]{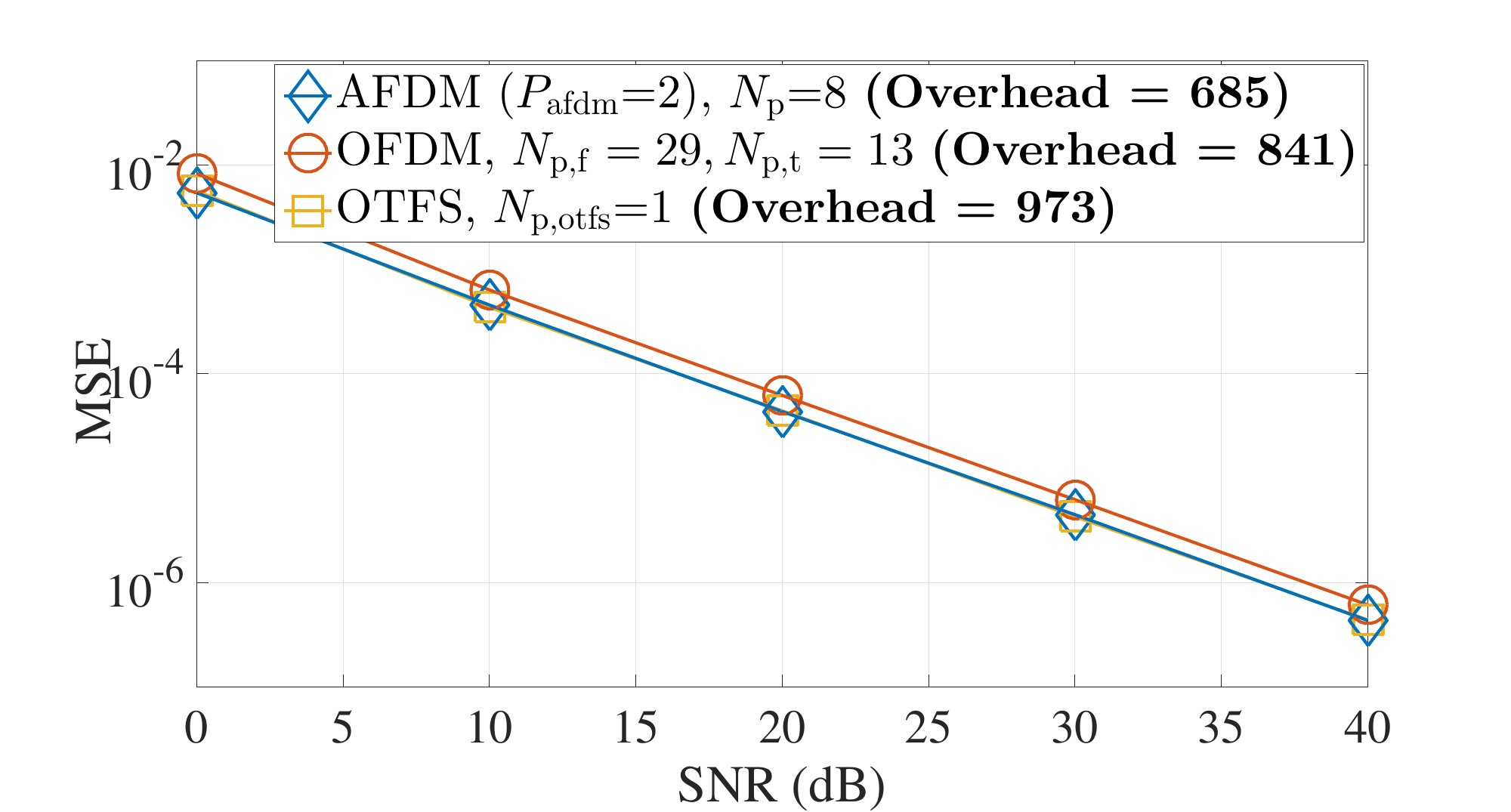} \\
  \small (a) $p_{\rm D}=0.2$\\
    \includegraphics[width=0.9\columnwidth]{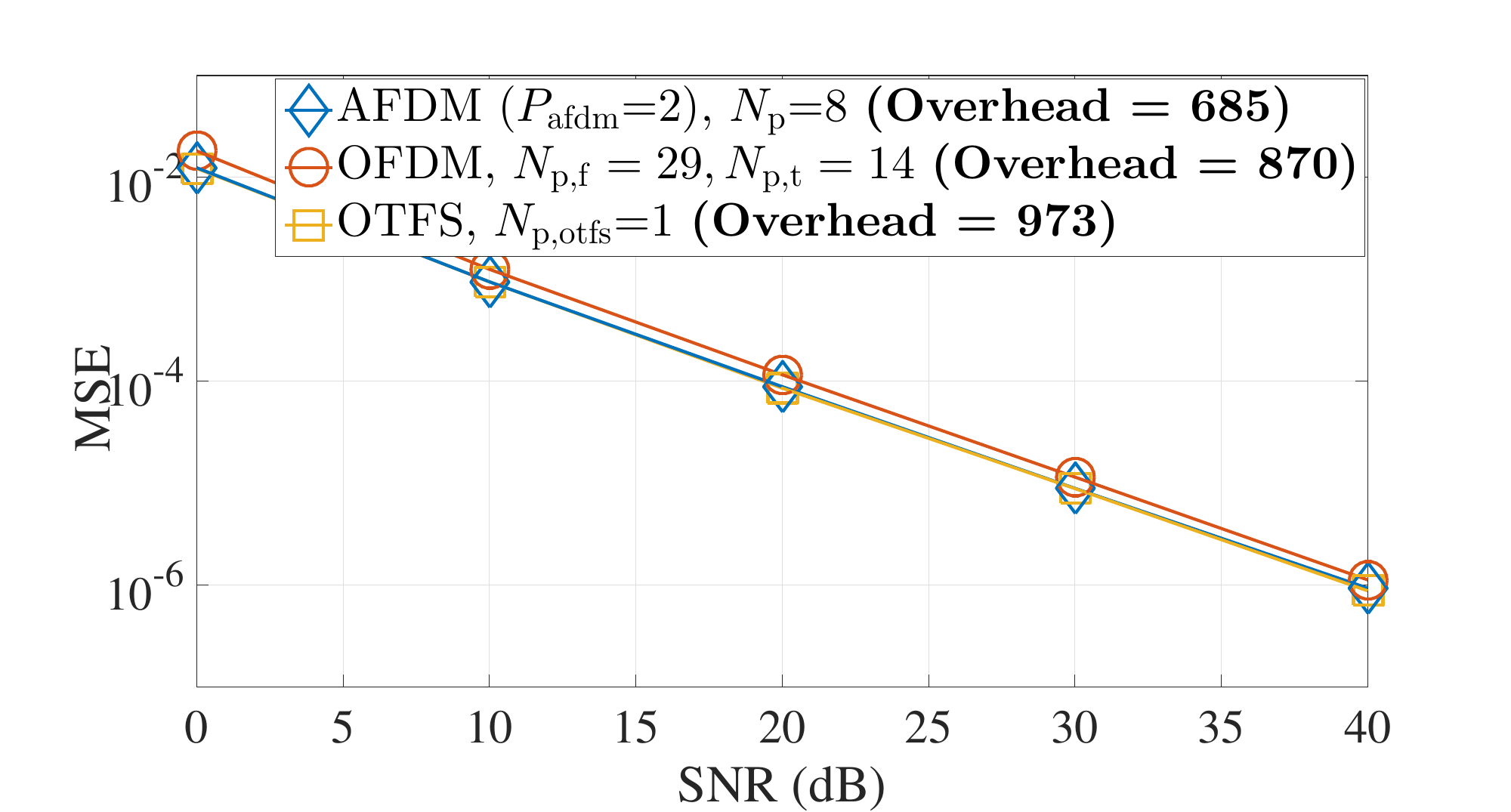} \\
      \small (b) $p_{\rm D}=0.4$
  \end{tabular}
  \caption{On-grid sparse-recovery MSE and pilot overhead for $N=4096, L=30, Q=7, p_{\rm d}=0.2$, $N_{\rm ofdm,symb}=16$, $N_{\rm otfs}=16$, $M_{\rm otfs}=256$ (overhead comparison done at comparable MSE levels for the three waveforms using \eqref{eq:ofdm_overhead}, \eqref{eq:afdm_overhead} and \eqref{eq:otfs_overhead})}
  \label{fig:simulations1_2}
\end{figure}

In Fig. \ref{fig:overhead_comparison}, we compare the pilot overhead (computed as in Fig. \ref{fig:simulations1_2}) needed by the three waveforms in four different configurations of the delay and Doppler sparsity level parameters, namely $(s_{\rm d},s_{\rm D})=(0.1,0.2)$, $(s_{\rm d},s_{\rm D})=(0.2,0.2)$, $(s_{\rm d},s_{\rm D})=(0.2,0.4)$ and $(s_{\rm d},s_{\rm D})=(0.4,0.2)$. Note that AFDM, configured with $P_{\rm afdm}=1$ for the first scenario and with $P_{\rm afdm}=2$ for the remaining ones, has the smallest overhead in all four scenarios. Also note that for OTFS, the overhead is the same in the four scenarios, as its pilot resources are determined as a function of the maximum delay and Doppler spreads, and not as a function of sparsity levels. This also hints at the possibility that the performance advantage over OTFS tends to diminish as the channel becomes less sparse. This intuition is validated by the figure, which shows a smaller gap in overhead values in the fourth sparsity scenario. The bit error rate (BER) and spectral efficiency (SE) of the considered waveforms in the first and fourth sparsity scenarios are given in Fig. \ref{fig:ber_se} assuming quadrature quadrature phase-shift keying (QPSK) data symbols. While OTFS achieves the best BER, as it is a full-diversity waveform, AFDM can achieve a very close BER performance by adjusting its chirp parameter $c_1=\frac{P_{\rm afdm}}{2N}$ to the sparsity level. However, the gap in performance is much wider in favor of AFDM when pilot overhead is considered while evaluating spectral efficiency.
\begin{figure}
    \centering
    \includegraphics[width=1.0\columnwidth]{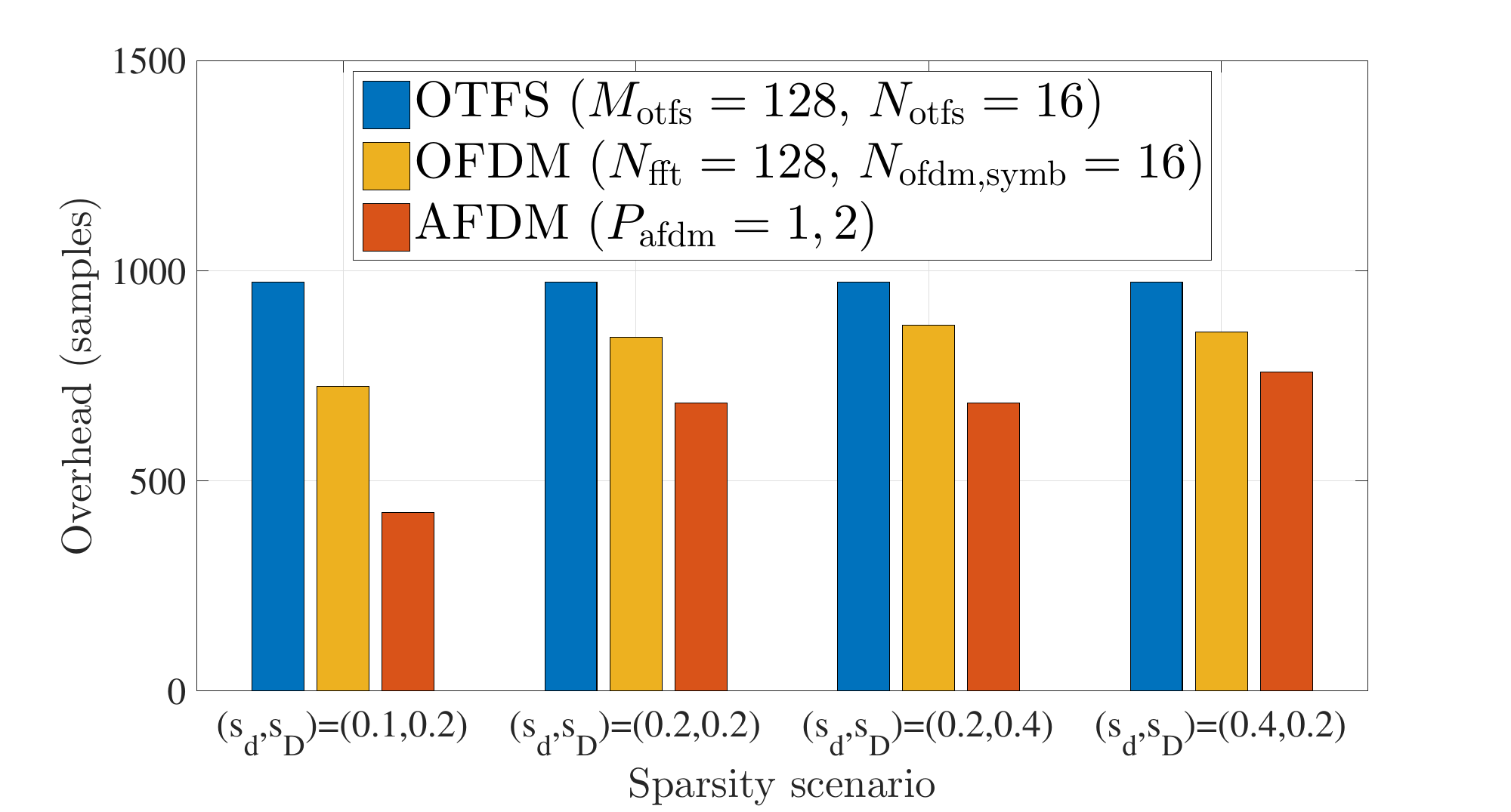}
    \caption{Pilot overhead for sparse recovery ($N=2048, L=30, Q=7$)}
\label{fig:overhead_comparison}
\vspace{-2mm}
\end{figure}
\begin{figure}
  \centering
  \begin{tabular}{ c  }
  \includegraphics[width=0.8\columnwidth]{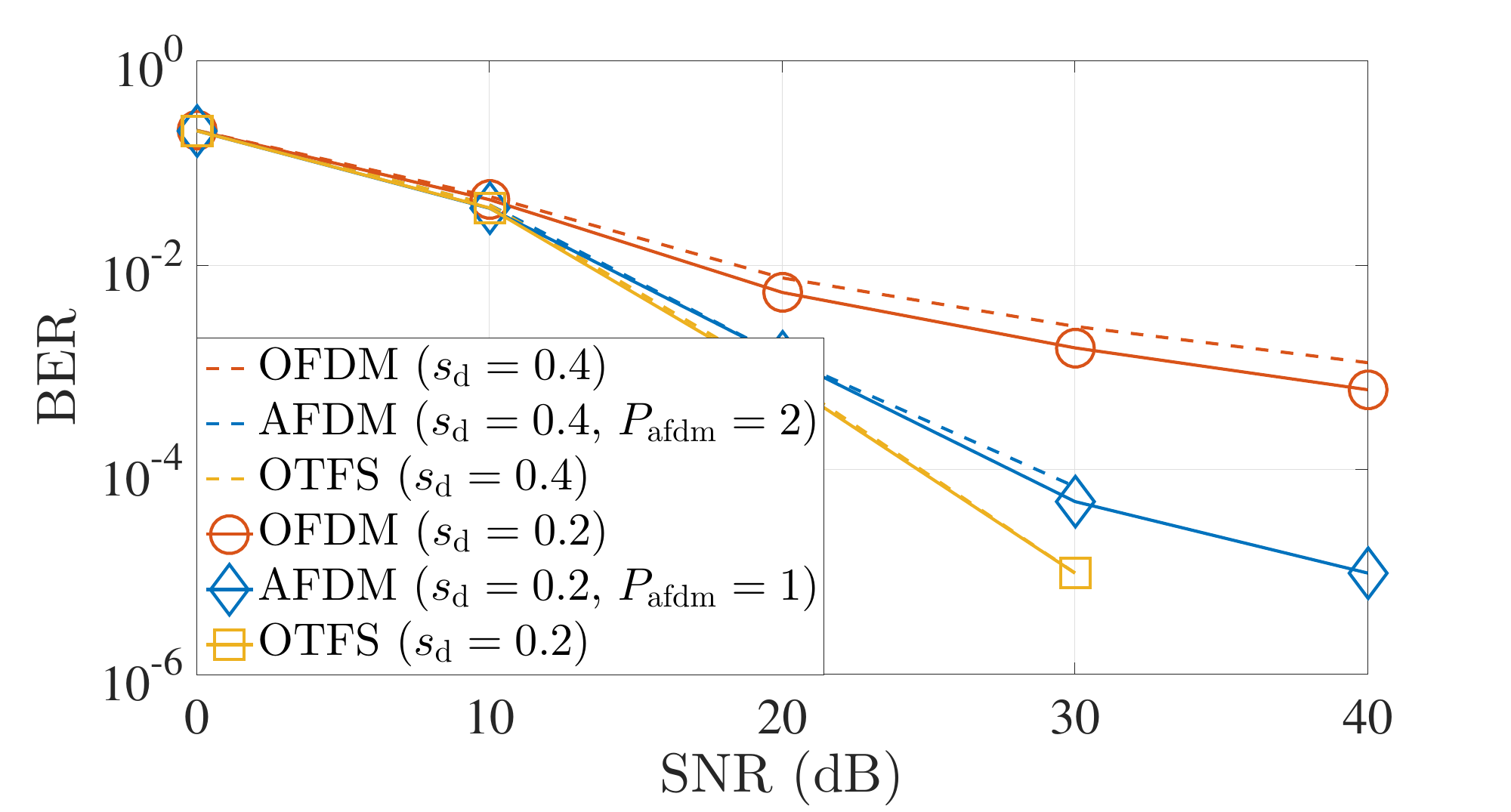} \\
  \small (a) BER \\
    \includegraphics[width=0.8\columnwidth]{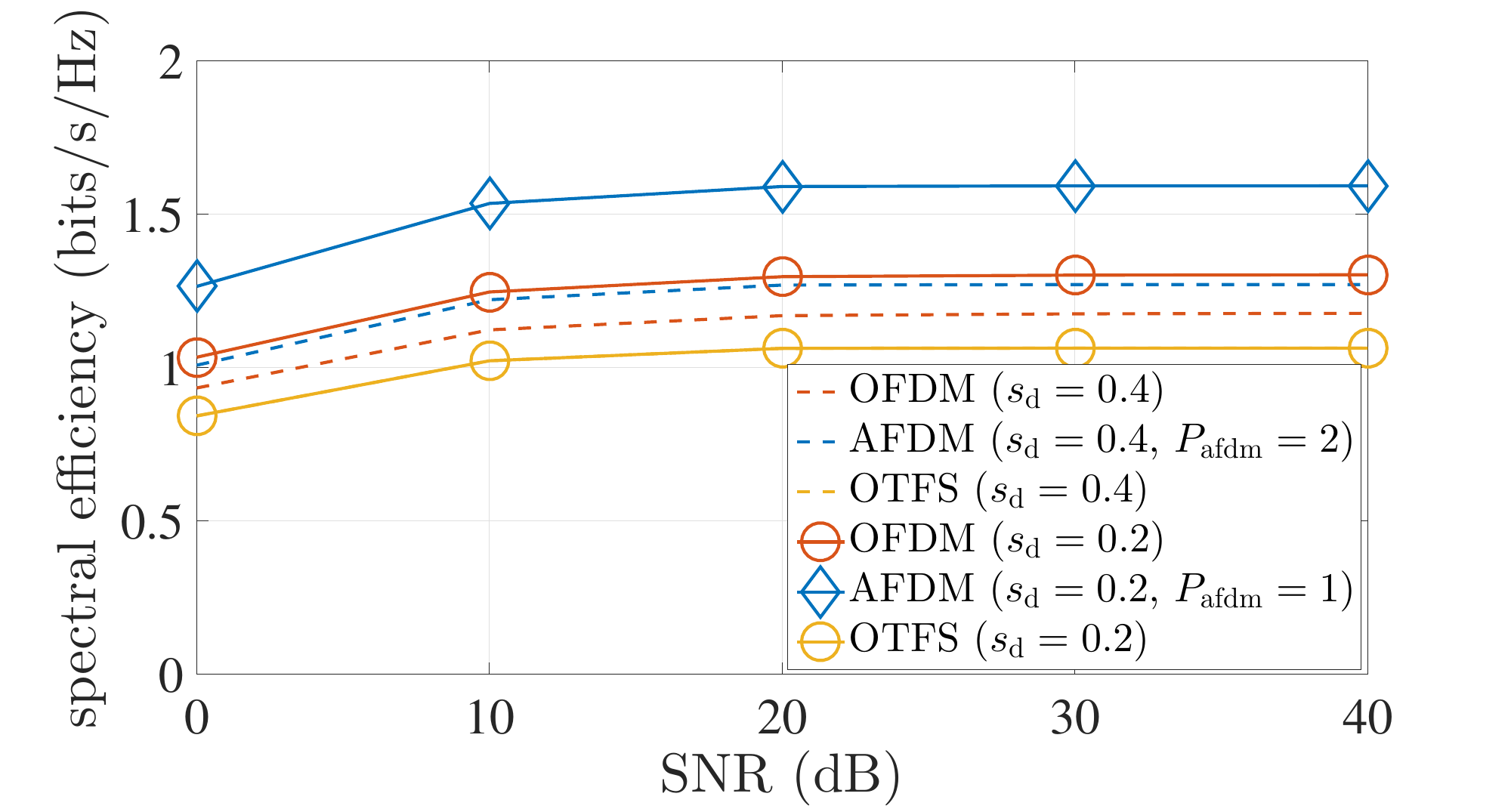} \\
      \small (b) SE
  \end{tabular}
  \caption{ BER and SE comparison ($N=2048, L=30, Q=7, p_{\rm D}=0.2$)}
  \label{fig:ber_se}
  \vspace{-2mm}
\end{figure}
 \begin{figure}
  \centering
  \begin{tabular}{ c  }
  \includegraphics[width=0.8\columnwidth]{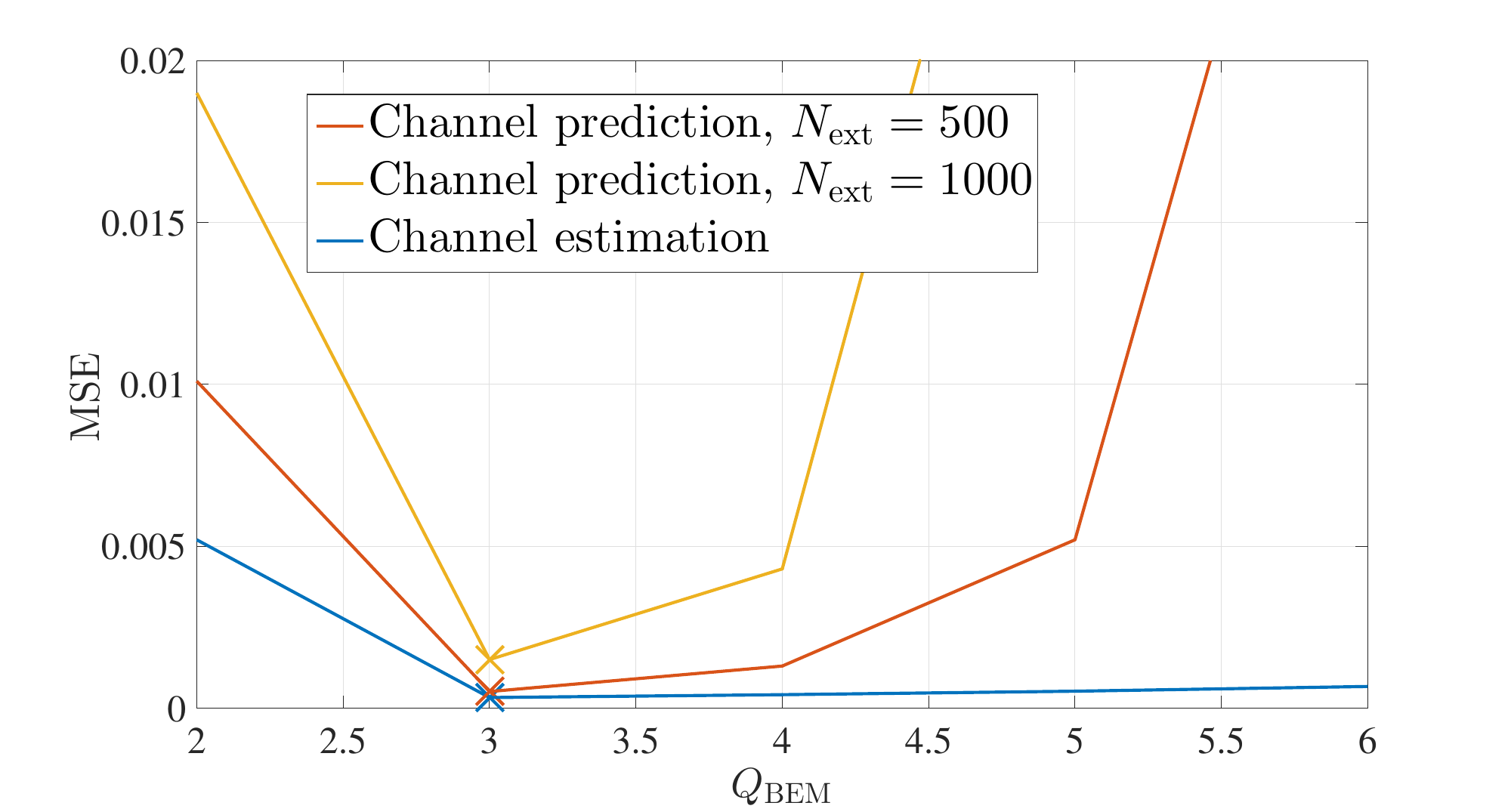} \\
  \small (a) $\text{SNR} = 10$ dB \\
    \includegraphics[width=0.8\columnwidth]{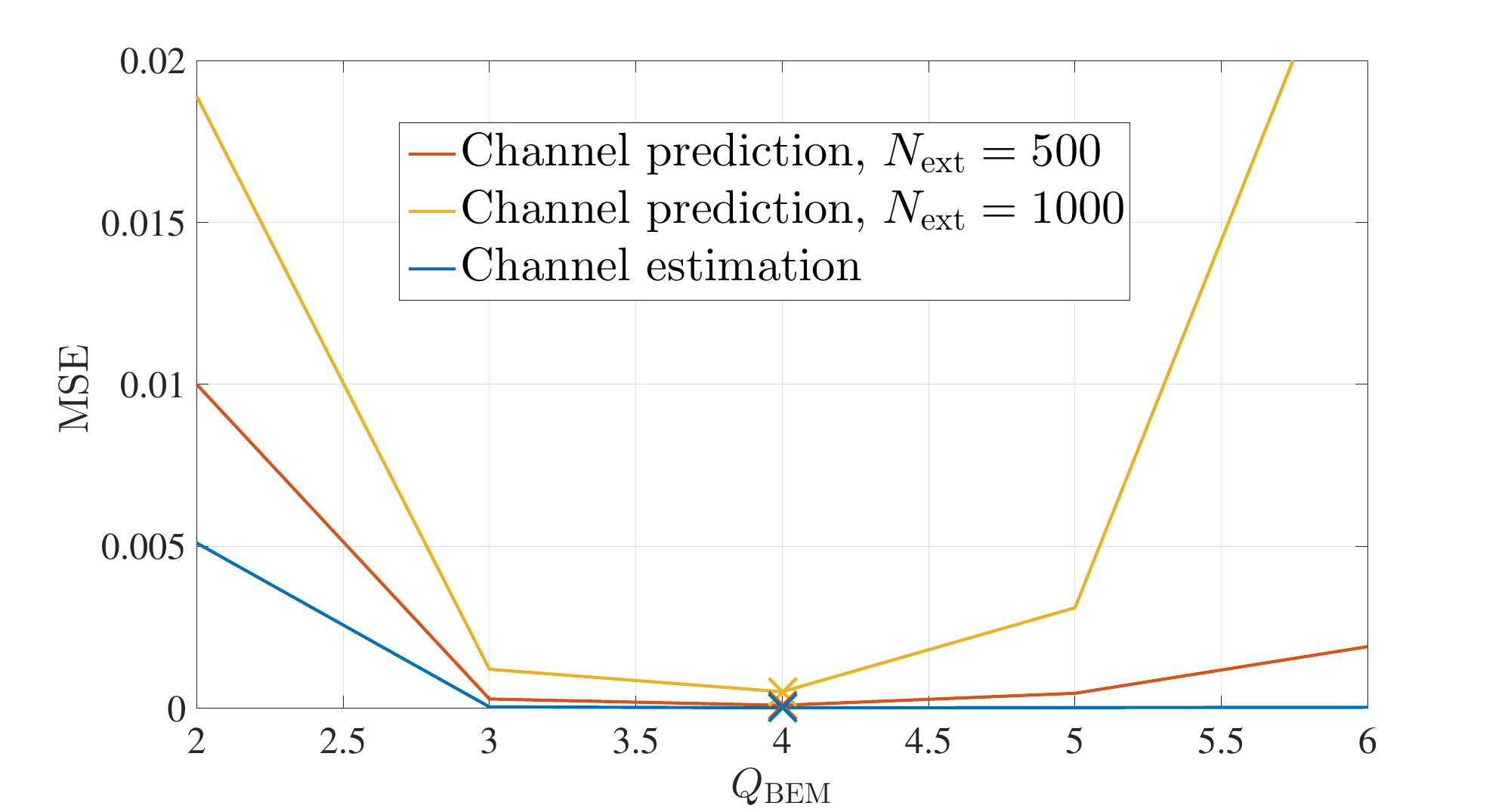} \\
      \small (b) $\text{SNR} = 40$ dB
  \end{tabular}
  \caption{Effect of parameter $Q_{\rm BEM}$ of the ``multiple shifted elementary BEMs'' method on the MSE performance of both off-grid LTV channel estimation and prediction ($N=2048, Q=7, N_{\rm p}=46$)}
  \label{fig:mse_qbem}
  \vspace{-2mm}
\end{figure}

We now turn to the estimation and extrapolation of DS-LTV channels with off-grid Doppler shifts, as modeled by Definition \ref{def:offgrid_dD_sparsity}. We start with the configuration of the parameters of our BEM model proposed in Subsection \ref{sec:multi_shifted_BEM}. The effect of the order $Q_{\rm BEM}$ of the elementary BEM on the MSE with different extrapolation horizons, namely $N_{\rm ext}=1000$, $N_{\rm ext}=500$ and $N_{\rm ext}=0$ (marked as ``channel estimation''), is shown in Figure \ref{fig:mse_qbem}. While increasing $Q_{\rm BEM}$ improves the precision of the model, it also increases the number of unknowns to be estimated, revealing the trade-off between model precision and pilot overhead. This is why $Q_{\rm BEM}=3$ yields the best MSE at the lower SNR value as opposed to $Q_{\rm BEM}=4$ at the higher. Note that the fact that both $Q_{\rm BEM}$ values are small confirms the result given in Theorem \ref{theo:bem_accuracy}.

Figure \ref{fig:simu2_offgrid} compares the MSE performance of the ``multiple shifted elementary BEMs'' method, with $Q_{\rm BEM} = 4$, to two other approaches, namely the multi-band single-BEM approach of \eqref{eq:dpss_onebem} (configured such that $\frac{\check{Q}_{\rm BEM}}{\textrm{\# active Doppler shifts per delay tap}}=4$), and a refined-grid approach \cite{Ganguly2019refinedgrid} (with a refinement factor of $O=4$). Note that these values of $Q_{\rm BEM}$, $\check{Q}_{\rm BEM}$, and $O$ are chosen to be equal to ensure a fair comparison across all methods, both in terms of model size (number of unknown parameters) and computational complexity. The comparison is conducted for both channel estimation (CE) and prediction, assuming $N_{\rm ext}=500$ and that the underlying waveform is AFDM with $P_{\rm afdm}=1$.
\begin{figure}
\centering
\includegraphics[width=1.0\columnwidth]{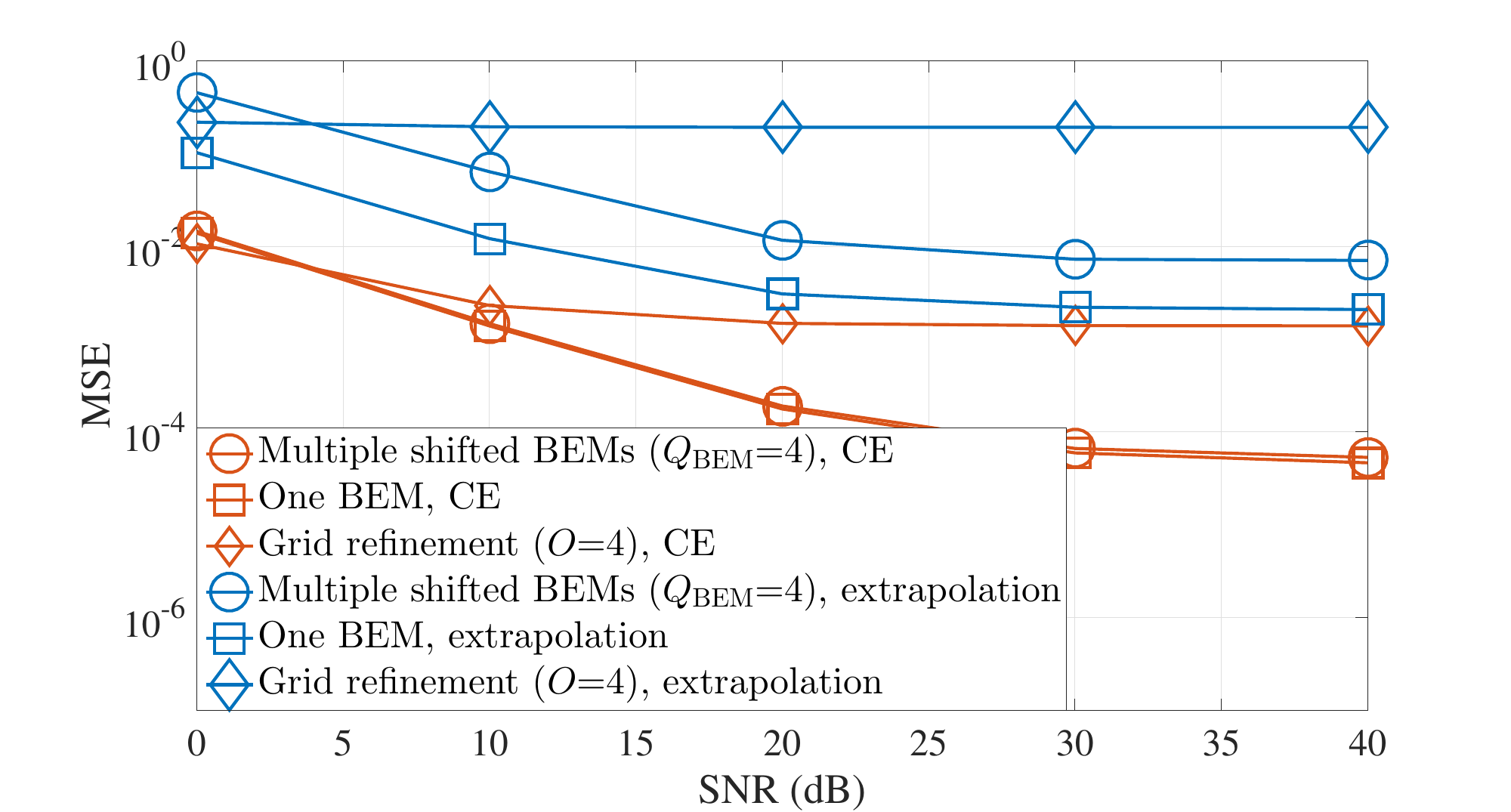}
\caption{MSE of different off-grid LTV channel estimation and prediction methods under AFDM ($N=2048, L=20, Q=7, p_{\rm d}=p_{\rm D}=0.2$)}
\label{fig:simu2_offgrid}
    \vspace{-2mm}
\end{figure}
For channel estimation, the multiple shifted BEMs approach and the multi-band single-BEM method achieve similar performance, both outperforming the refined-grid technique (due to the latter being ill-conditioned since $N_{\rm D}>1$ as explained at the beginning of Section \ref{sec:off_grid}). However, for channel prediction, our approach exhibits a slight performance degradation compared to the single BEM method. Nevertheless, as highlighted in Remark \ref{rem:DPSS_prediction}, our method requires a significantly smaller codebook size than the multi-band single-BEM approach shown by Table \ref{tab:codebook} (for the same setting as in Figures \ref{fig:simu1_offgrid} and \ref{fig:simu2_offgrid}), making it more efficient in practical implementations.
\begin{table}
\begin{center}
    \caption{Size of codebook for each method ($L=20, Q=7, p_{\rm d}=0.2, p_{\rm D}=0.2$)}
    \label{tab:codebook} 
    \begin{tabular}[c]{||p{3.75cm}||p{3.5cm}||} 
\hline
    \textbf{Method} & \textbf{Codebook Size} \\
\hline
\hline
    Multi-band single BEM \cite{channel_extension}. Number of Doppler bands = $(2Q+1)p_{\rm D}=3$, $\check{Q}_{\rm BEM}=12$          & 2,621,360          \\
    Multiple shifted BEMs \eqref{eq:ch_model_off_grid}. Number of elementary BEMs = $(2Q+1)p_{\rm D}=3$, $Q_{\rm BEM}=4$    & 1,200              \\
    \hline
    \end{tabular}
    \end{center}
        \vspace{-2mm}
\end{table}

To test the effect of mobility on the proposed channel estimation and prediction methods when using AFDM, we use two maximum Doppler shift values: $Q=3$ and $Q=7$, corresponding to maximum speeds of 170 km/h and 396 km/h, respectively. Doppler sparsity is $p_{\rm D}=0.4$ in the lower-mobility case and $p_{\rm D}=0.2$ in the other case, ensuring that the number of active Doppler clusters is the same in both scenarios. As expected, the MSE performance is better in the lower-mobility case. However, the degradation in performance is limited when the maximum speed is almost doubled, demonstrating the robustness of both AFDM and the proposed estimation and prediction schemes under high mobility.
\begin{figure}
\centering
\includegraphics[width=1.0\columnwidth]{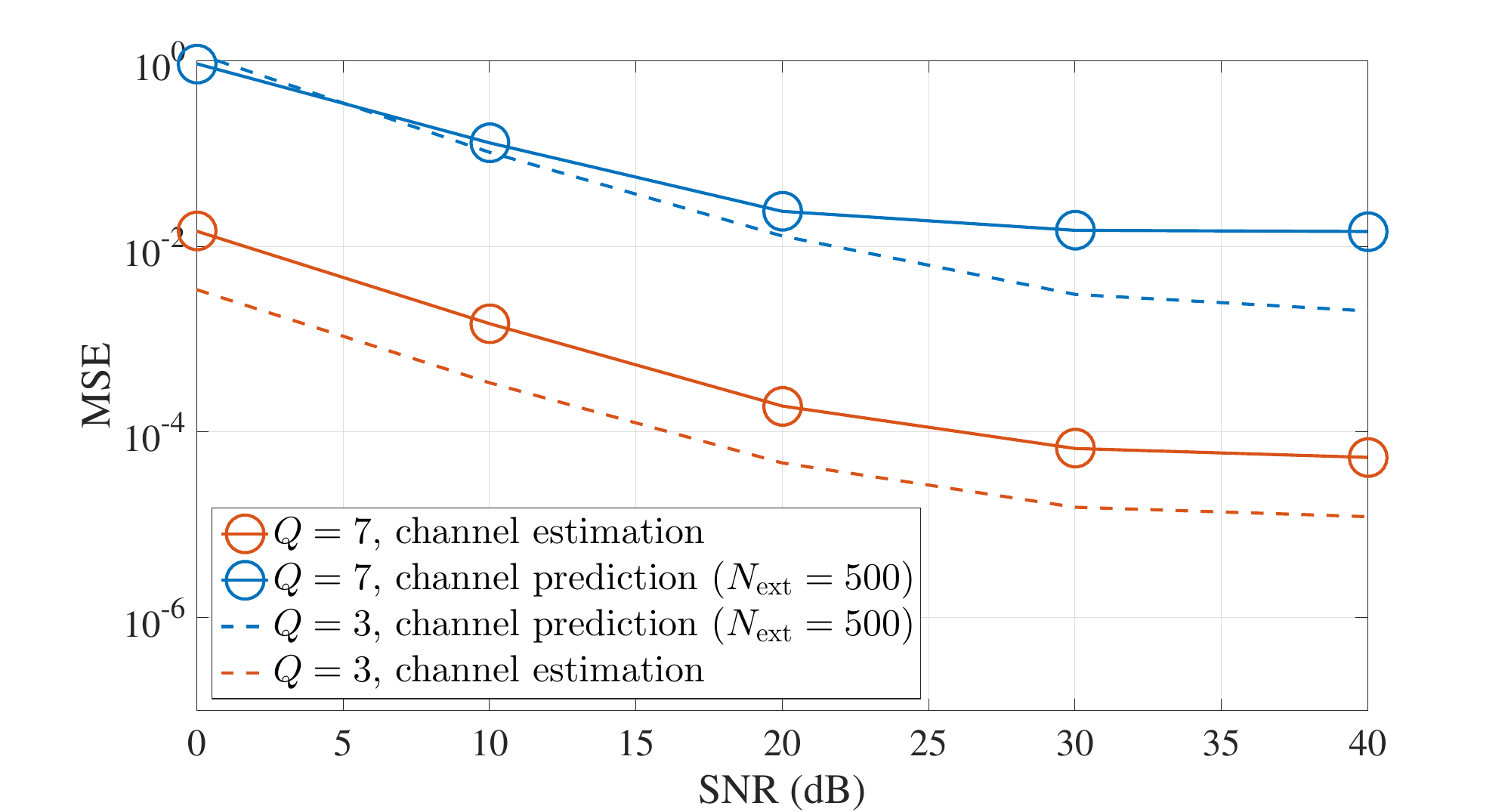}
\caption{Effect of mobility (value of $Q$) on MSE performance of off-grid LTV channel estimation and prediction using the ``multiple shifted elementary BEMs'' and DPSS extrapolation approaches under AFDM ($N=2048, L=20, Q_{\rm BEM}=4, p_{\rm d}=0.2$)}
\label{fig:simu3_offgrid}
    \vspace{-2mm}
\end{figure}

Finally, we demonstrate that the superiority of AFDM in terms of sparse-recovery performance, as established using compressed-sensing tools under the on-grid model, extends to the off-grid model and LMMSE channel estimation. We used 100 realizations of channels having a off-grid Type-1 delay-Doppler sparsity with $p_{\rm d}=0.2$, $p_{\rm D}=0.2$ and $N=2048, N_{\rm ofdm,symb}=32, L=20, Q=7$.
As highlighted in Remark \ref{rem:waveform}, Fig. \ref{fig:simu1_offgrid} illustrates that when dealing with off-grid Doppler shifts, AFDM maintains the superiority that was proven with the on-grid model. Indeed, the values of $N_{\rm p}$, $N_{\rm ofdm,symb}$, $N_{\rm p,t}$ and $N_{\rm p,f}$ used to generate the figure, while yielding almost identical pilot overheads for OFDM and AFDM ($\textrm{overhead}_{\rm AFDM}=766$, $\textrm{overhead}_{\rm OFDM}=788$, result in a superior MSE performance for AFDM in both channel estimation and channel prediction tasks.
\begin{figure}
     \centering
    \includegraphics[width=1.0\columnwidth]{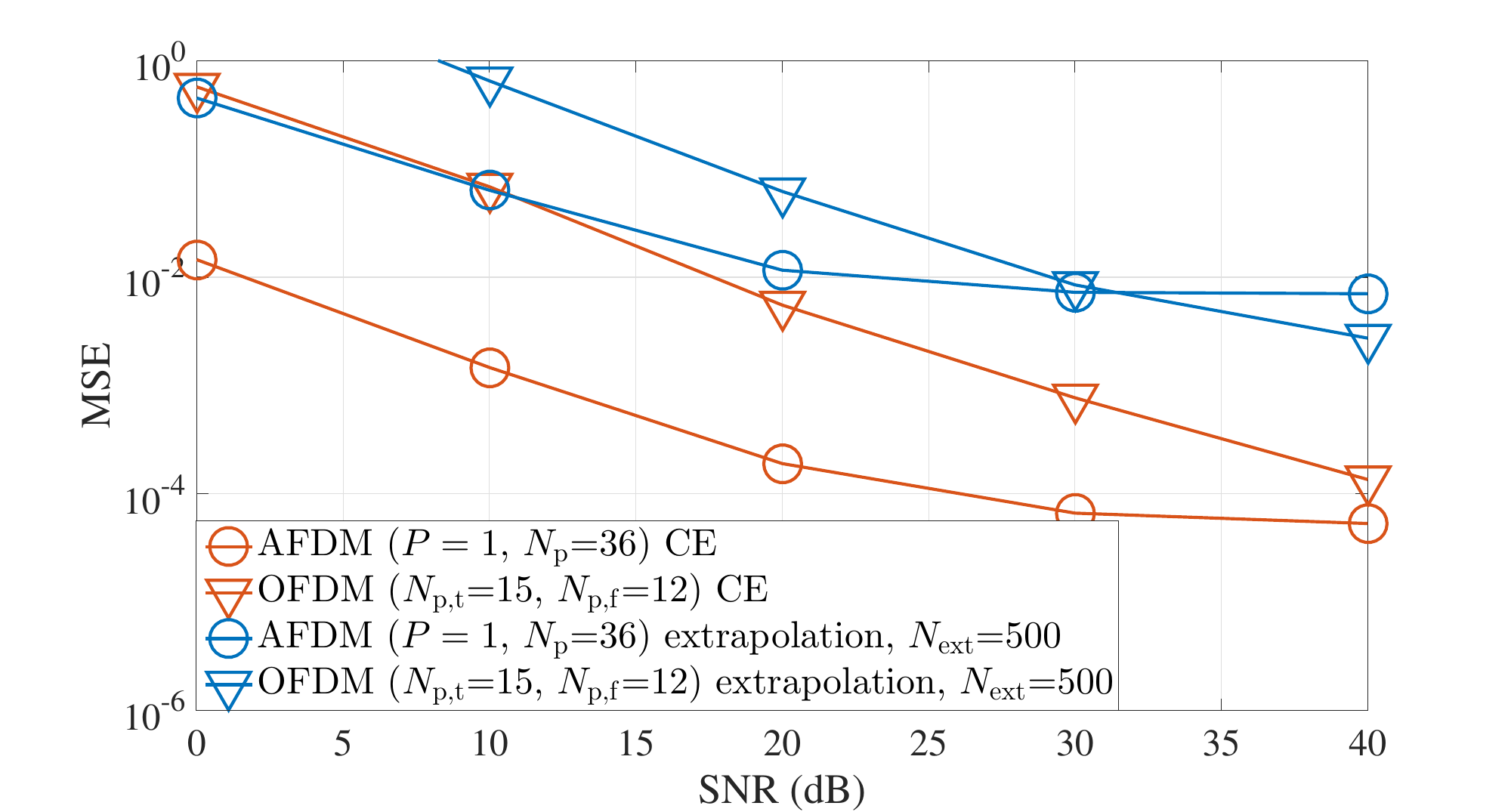}
    \caption{MSE of off-grid LTV channel estimation and prediction using AFDM and OFDM for $N=2048, L=20, Q=7, p_{\rm d}=p_{\rm D}=0.2$,$N_{\rm ofdm,symb}=32$, $N_{\rm fft}=64$ (MSE comparison done at comparable overhead levels for the two waveforms)}
        \label{fig:simu1_offgrid}
        \vspace{-2mm}
    \end{figure}

\section{Conclusions}

This paper investigated channel estimation for doubly sparse LTV channels, addressing both on-grid and off-grid Doppler shifts. The on-grid approximation, through its link to the hierarchical-sparsity compressed-sensing paradigm, was used to provide a rigorous performance comparison when channel estimation is conducted based on different underlying waveforms. This theoretical comparison established that the AFDM waveform can achieve the same sparse-recovery guarantees while requiring reduced pilot overhead compared to SCM, OFDM, and OTFS.
To relax the on-grid assumption, we introduced a DPSS-based channel representation model consisting of multiple frequency-shifted elementary BEMs capable of handling off-grid Doppler effects, and we complemented it with analytical precision guarantees. We next used the proposed channel representation to conceive a practical LMMSE estimator for LTV channels. Additionally, a channel predictor was introduced by taking advantage of the DPSS extrapolation capability of the proposed LMMSE estimator. Numerical results confirmed AFDM superiority in both estimation and prediction under delay-Doppler sparsity, making it a strong candidate for future wireless systems.

\appendices

\section{Proof of Theorem \ref{theo:HiRIP_SCM_OFDM}}
\label{app:proof_theo_scm_ofdm}
In the case of SCW, defining $\tilde{\alpha}_{l,q}\triangleq\alpha_{l,q}e^{\imath2\pi\frac{lq}{N}}$, rearranging $\mathbf{y}_{\rm p}$ into $\tilde{\mathbf{y}}_{\rm p}$ (composed of $L$ successive blocks, with the $l$-th blocks composed of the $l$-th sample in each of the $N_{\rm p}$ pilot intervals) and assuming $2Q+1$ divides $N$ and that the $p$-th pilot position (for any $p\in\Iintv{1,N_{\rm p}}$) satisfies $m_p=q_p\frac{N}{2Q+1}$ for some $q_p\in\Iintv{0,2Q}$, we can write $\tilde{\mathbf{y}}_{\rm p}=\widetilde{\mathbf{M}}_{\rm p}^{\rm scw}\tilde{\boldsymbol{\alpha}}$
where $\tilde{\boldsymbol{\alpha}}$ is the vectorized form of $\tilde{\alpha}_{l,q}$ and
\begin{equation}
    \widetilde{\mathbf{M}}_{\rm p}^{\rm scw}\triangleq\mathbf{I}_{L}\otimes\left(\mathrm{diag}\left(\mathrm{p}_1,\ldots,\mathrm{p}_{N_{\rm p}}\right)\overline{\mathbf{F}}_{2Q+1,N_{\rm p}}\right) 
\end{equation}
Here, $\overline{\mathbf{F}}_{2Q+1,N_{\rm p}}$ is the partial inverse Fourier measurement matrix formed from $N_{\rm p}$ rows of the $(2Q+1)$-point inverse DFT matrix. The HiRIP of $\widetilde{\mathbf{M}}_{\rm p}^{\rm scw}$ can thus be derived and proven to be equal to the value given in the theorem statement by using the known RIP of partial inverse Fourier measurement matrices \cite[Theorem 4.5]{partial_fourier} followed by applying \cite[Theorem 4]{hierarchical} pertaining to the HiRIP of hierarchical measurement matrices having the Kronecker property. This completes the part of the proof related to SCM. 
As for OFDM, the measurement matrix of the estimation problem is 
\begin{equation}
 \widetilde{\mathbf{M}}_{\rm p}^{\rm ofdm}\triangleq(\mathrm{diag}(\mathrm{p}_1,\ldots,\mathrm{p}_{N_{\rm p,f}})\mathbf{F}_{L,N_{\rm p,f}})\otimes\overline{\mathbf{F}}_{2Q+1,N_{\rm p,t}}
\end{equation}
where $\mathbf{F}_{L,N_{\rm p,f}}$ is the partial Fourier measurement matrix formed from $N_{\rm p,f}$ rows of the $L$-point inverse DFT matrix. The HiRIP of $\widetilde{\mathbf{M}}_{\rm p}^{\rm ofdm}$ thus follows from the RIP of the partial Fourier measurement matrix and the HiRIP result pertaining to Kronecker hierarchical measurements.

\section{Proof of Theorem \ref{theo:HiRIP_AFDM}}
\label{app:proof_theo}
First, out of the pilot samples set $\mathcal{P}$, consider the subset $\mathcal{P}_p$ associated with the $p$-th pilot symbol transmitted at the DAFT index $m_p$ (Fig. \ref{fig:AFDM_pilot_pattern}). To homogenize the sensing signal model associated with edge samples and inner samples of $\mathcal{P}_p$, we apply two overlap-add operations: adding the samples received within the index interval $\Iintv{m_p-Q,m_p-1}$ to those received within $\Iintv{m_p+(L-1)P_{\rm afdm}-Q,m_p+(L-1)P_{\rm afdm}-1}$ and the samples received within $\Iintv{m_p+(L-1)P_{\rm afdm}+1,m_p+(L-1)P_{\rm afdm}+Q}$ to those received within $\Iintv{m_p+1,m_p+Q}$.
Now, define
\begin{equation}
\label{eq:Dl}
    \mathcal{D}_l\triangleq\left\{(\tilde{l},q) \mathrm{s.t.}(q+P_{\rm afdm}\tilde{l})_{(L-1)P_{\rm afdm}+1}=l\right\}
\end{equation}
as the set of delay-Doppler grid points that potentially contribute to the pilot sample received at DAFT domain index $l\in\Iintv{m_p,m_p+(L-1)P_{\rm afdm}}$ (Fig. \ref{fig:interval_k}) after the two overlap-add operations described above. Note that $\mathcal{D}_l$ does not depend on the pilot symbol index $p$ and that it has a cardinality that does not change with $l$ and which satisfies $\left|\mathcal{D}_l\right|\leq 2\lceil \frac{Q}{P_{\rm afdm}}\rceil+1$. Next, define $\boldsymbol{\alpha}_{\mathcal{D}_l}\triangleq\left[\alpha_{l,q}\right]_{(l,q)\in\mathcal{D}_l}$ and
$\widetilde{\boldsymbol{\alpha}}\triangleq[\boldsymbol{\alpha}_{\mathcal{D}_0}^{\rm T}\quad\cdots\quad\boldsymbol{\alpha}_{\mathcal{D}_{(L-1)P_{\rm afdm}}}^{\rm T}]^{\rm T}$. The entries of $\widetilde{\boldsymbol{\alpha}}$ are just a permutation of $\boldsymbol{\alpha}$ and estimating one of these vectors directly provides an estimate of the other. Now, it can be shown that when we set $P_{\rm afdm}$ as in the theorem and $\epsilon>0$ as small as needed, then $\widetilde{\boldsymbol{\alpha}}$ is $\left(\tilde{s}_{\rm d},\tilde{s}_{\rm D}\right)$-hierarchically sparse with high probability:
\begin{figure}
  \centering
  \begin{tabular}{ c @{\hspace{5pt}} c }
  \includegraphics[width=.48\columnwidth]{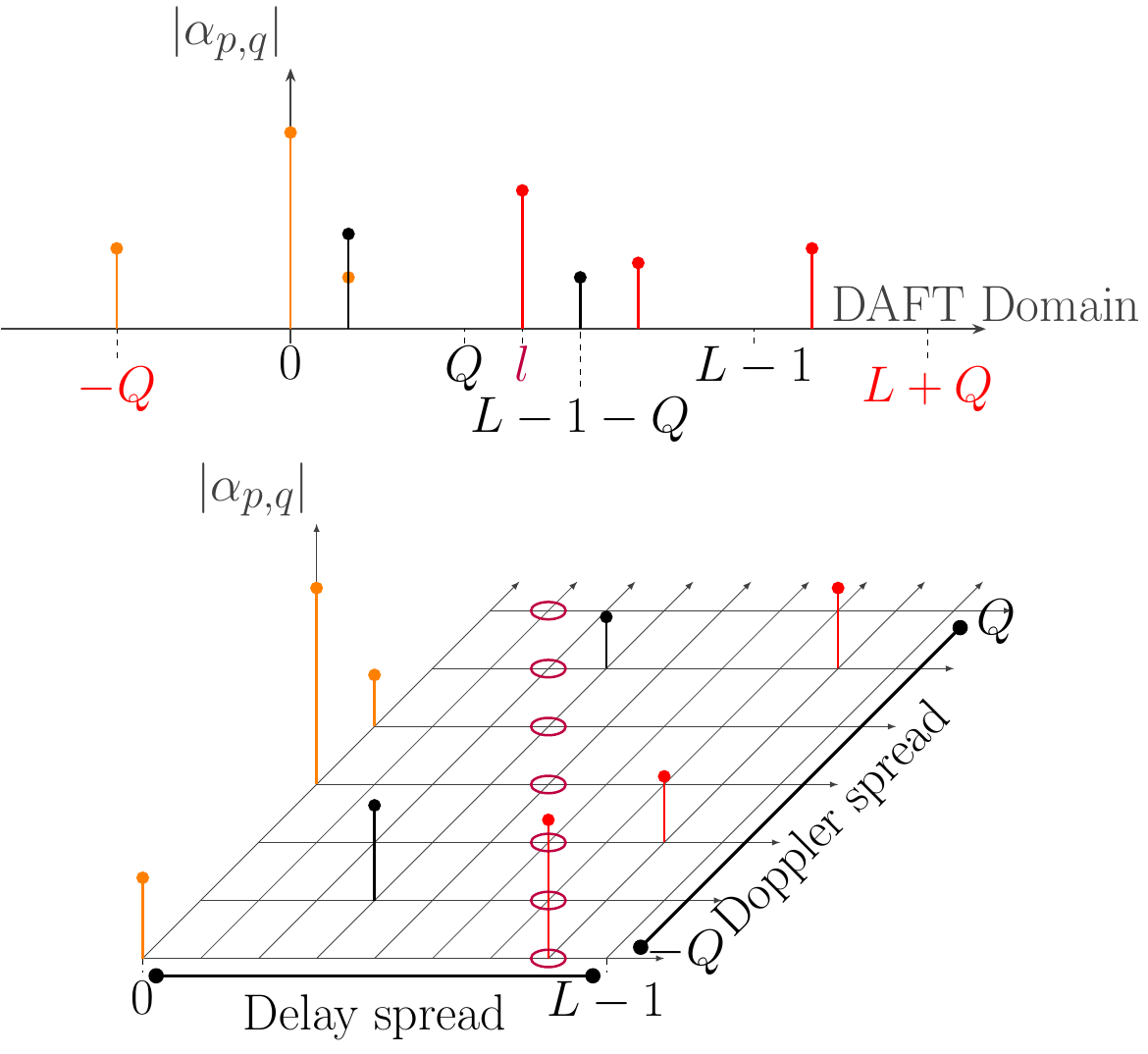} &
    \includegraphics[width=.48\columnwidth]{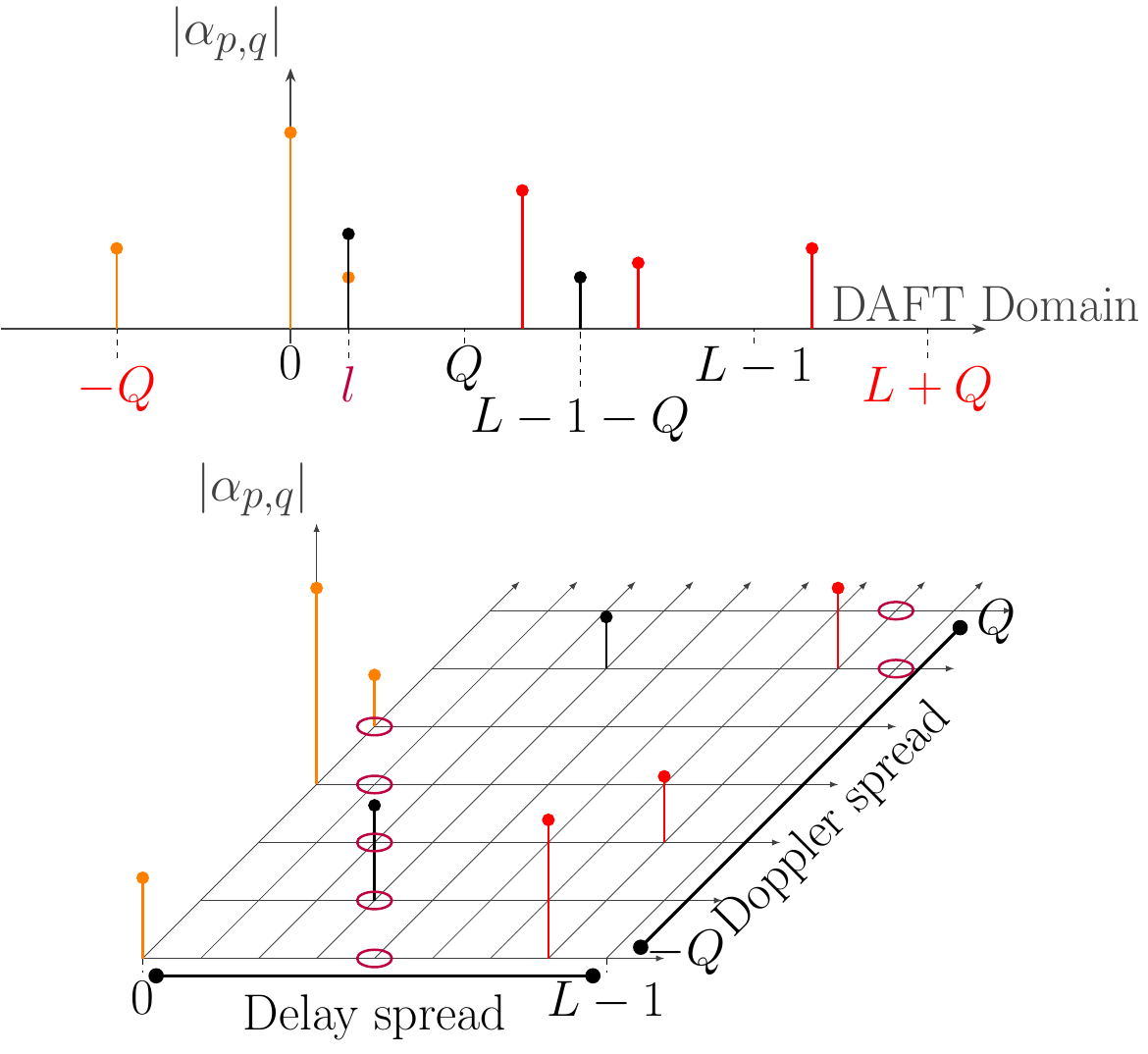} \\
    \small (a) &
      \small (b)
  \end{tabular}
  \medskip
  \vspace{-4mm}
  \caption{Two examples of the set $\mathcal{D}_l$ ($P_{\rm afdm}=1$, $m_p=0$) (a) for an $l$ resulting in a whole diagonal, (b) for an $l$ resulting in a wrapped diagonal. The grid points forming $\mathcal{D}_l$ are shown surrounded by red rings.}
  \label{fig:interval_k}
  \vspace{-5mm}
\end{figure}
\begin{equation}
    \label{eq:tilde_s_D}
    \tilde{s}_{\rm d}=(L-1)P_{\rm afdm}+1,\quad\tilde{s}_{\rm D}=(1+\epsilon)\log(LP_{\rm afdm})\:.
\end{equation}
Indeed, the first level (of size $(L-1)P_{\rm afdm}+1$) of $\widetilde{\boldsymbol{\alpha}}$ is sensed without compression, with a number of measurements equal to $(L-1)P_{\rm afdm}+1$ while $\tilde{s}_{\rm D}$ can be determined thanks to Definition \ref{def:dD_sparsity} and Assumptions \ref{assum:technical} and \ref{assum:technical_diag} and applying the same approach as in the proof of Lemma \ref{lem:DS_HS} to $\Tilde{S}_{{\rm D},l}\triangleq\sum_{(\tilde{l},q)\in\mathcal{D}_l}I_{\tilde{l},q}$.
Now, we can write the signal model of sensing $\widetilde{\boldsymbol{\alpha}}$ as
\begin{equation}
    \label{eq:tilde_yp}
    \widetilde{\mathbf{y}}_{\rm p}=\widetilde{\mathbf{M}}_{\rm p}\widetilde{\boldsymbol{\alpha}}
    +\widetilde{\mathbf{w}}_{\rm p},
\end{equation}
where $\widetilde{\mathbf{y}}_{\rm p}=[\widetilde{\mathbf{y}}_{{\rm p},0}^{\rm T}\quad\cdots\quad\widetilde{\mathbf{y}}_{{\rm p},(L-1)P_{\rm afdm}}^{\rm T}]^{\rm T}$. For each $l$, $\widetilde{\mathbf{y}}_{{\rm p},l}$ is a $N_{\rm p}\times 1$ vector composed of the pilot samples received at DAFT domain positions $\left\{m_p+l\right\}_{p=1\cdots N_{\rm p}}$. Note that by this definition $\Tilde{\mathbf{y}}_{\rm p}$ is obtained by permuting $\mathbf{y}_{\rm p}$ in \eqref{eq:yp} in accordance with the permutation that gives $\Tilde{\boldsymbol{\alpha}}$ from $\boldsymbol{\alpha}$. Next, by setting for each $p\in\Iintv{1,N_{\rm p}}$ $m_p=q_p\frac{N}{2\lceil \frac{Q}{P_{\rm afdm}}\rceil+1}$ for some integer $q_p$ it follows from \eqref{eq:y_output_integer} and \eqref{eq:Dl} that $\widetilde{\mathbf{M}}_{\rm p}$ has the following Kronecker structure
\begin{eqnarray}
    \label{eq:kron}
    \widetilde{\mathbf{M}}_{\rm p}=\mathbf{I}_{(L-1)P_{\rm afdm}+1}\otimes\widetilde{\mathbf{M}}_{\cal D},\\
    \widetilde{\mathbf{M}}_{\mathcal{D}} = \mathrm{diag}({\rm p}_1\cdots{\rm p}_{N_{\rm p}})\mathbf{F}_{2\lceil\frac{Q}{P_{\rm afdm}}\rceil+1,N_{\rm p}}\boldsymbol{\Psi}\:,
\end{eqnarray}
where $\mathbf{F}_{2\lceil\frac{Q}{P_{\rm afdm}}\rceil+1,N_{\rm p}}$ is a $\left(2\lceil\frac{Q}{P_{\rm afdm}}\rceil+1\right)\times N_{\rm p}$ partial Fourier measurement matrix, and $\boldsymbol{\Psi}$ is a diagonal matrix with unit-modulus entries. We can thus use \cite[Theorem 4.5]{partial_fourier} pertaining to sub-sampled Fourier matrices to conclude that, for sufficiently large $L$, $Q$, sufficiently small $\delta$, and
\begin{equation}
    \label{eq:cond_N_p}
    N_{\rm p}>\Omega(\frac{1}{\delta^{2}}\log^2\frac{1}{\delta}\log\frac{\log (LP_{\rm afdm})}{\delta}\log(LP_{\rm afdm})\log\frac{Q}{P_{\rm afdm}})
\end{equation}
the RIP constant $\delta_{\tilde{s}_{\rm D}}$ of $\widetilde{\mathbf{M}}_{\mathcal{D}}$ satisfies $\delta_{\tilde{s}_{\rm D}}\leq\delta$ with probability $1-e^{-\Omega \left(\log{\frac{Q}{P_{\rm afdm}}}\log{\frac{1}{\delta}}\right)}$. The RIP of $\mathbf{I}_{(L-1)P_{\rm afdm}+1}$ trivially satisfies $\delta_{\tilde{s}_{\rm d}}=0$. As for the HiRIP of $\widetilde{\mathbf{M}}_{\rm p}$, we can apply \cite[Theorem 4]{hierarchical} to \eqref{eq:kron} thanks to its Kronecker structure to conclude that, if $N_{\rm p}$ and $\delta$ are as in \eqref{eq:cond_N_p}, then
\begin{equation}
    \delta_{s_{\rm d},s_{\rm D}}\leq\delta_{\tilde{s}_{\rm d}}+\delta_{\tilde{s}_{\rm D}}+\delta_{\tilde{s}_{\rm d}}\delta_{\tilde{s}_{\rm D}}\leq\delta\:.
\end{equation}
This completes the proof of the theorem.

\section{Proof of Theorem \ref{theo:bem_accuracy}}
\label{app:bem_accuracy}
We use the following lemma taken from \cite[Theorem 2.4]{MSE_DPSS}. 
\begin{lemma}{\cite[Theorem 2.4]{MSE_DPSS}}
    \label{lem:MSE_DPSS}
    Let $h(t)$ be a continuous-time zero-mean wide-sense stationary random process with power spectrum $P_{h}(f)=\frac{1}{B}\mathrm{rect}(\frac{f-F_{\rm c}}{2B})$.
Denote by $\mathbf{h}=[h(0T_{\rm s})\,\ldots\:h((N-1)T_s)]^{\rm T}$ a vector of samples acquired from $h(t)$ with a sampling period $T_s \leq \frac{1}{2F_{\rm c}+B}$. Let $W=\frac{B T_s}{2}$, $\mathbf{U}_k$ the matrix form of the $k$ first $(N,W)$-DPSS vectors, $\mathbf{E}_{f}\triangleq\mathrm{diag}\left( e^{\imath 2\pi f 0} \ldots e^{\imath 2\pi f(N-1)}\right)$ and $\mathbf{P}_{k}\triangleq\mathbf{E}_{F_{\rm c}T_{\rm s}}\mathbf{U}_k\mathbf{U}_k^{\rm H}\mathbf{E}_{F_{\rm c}T_{\rm s}}^{\rm H}$. Then
$\mathbb{E}[\left\|\mathbf{h}-\mathbf{P}_{k}\mathbf{h}\right\|_2^2]=\frac{1}{2W} \sum_{l=k}^{N-1} \lambda_l^{(N,W)}$.
\end{lemma}

In what follows, we use Lemma \ref{lem:MSE_DPSS} to upper bound $\mathbb{E}[|h_{l,q,n}-h_{l,q,n}^{\rm BEM}|^2]$ by first
rewriting $h_{l,q,n}$ as the sampled version of the continuous-time signal $h_{l,q}(t)$ defined as
\begin{equation}
    \label{eq:continuous_channel}
    h_{l,q}(t) \triangleq \sum_{i=1}^{N_{\rm D}} \alpha_{l,q,i} e^{\imath 2\pi f_it}, \quad t \in \mathbb{R}.
\end{equation}
with \(f_i\triangleq\frac{q}{NT_{\rm s}}+\frac{\kappa_i}{NT_{\rm s}}\). To prove that the PSD of the random process $h_{l,q}(t)$ has the desired property, we derive its autocorrelation function $R_{h_{l,q}}(\tau) \triangleq \mathbb{E}[h_{l,q} (t)h_{l,q}^*(t+\tau)]$
\begin{align}
    \label{eq:autocorr1}
    R_{h_{l,q}}(\tau) &=\sum_{i=1}^{N_D} \sum_{j=1}^{N_D} \mathbb{E}[\alpha_{l,q,i} \alpha_{l,q,j}^*] \, 
\mathbb{E}[e^{\imath 2\pi f_i t} e^{-\imath 2\pi f_j(t+\tau)}]\nonumber\\
&=\sum_{i=0}^{N_D} \sigma_\alpha^2 \, \mathbb{E}\left[e^{-\imath 2\pi f_i \tau}\right]
\end{align}
where the first equality follows from \eqref{eq:continuous_channel} and the second from $f_i \sim \mathcal{U}([\frac{q}{N T_{\rm s}} - \frac{1}{2 N T_{\rm s}}, \frac{q}{N T_{\rm s}} + \frac{1}{2 N T_{\rm s}}])$, $\alpha_{l,q,i} \sim \mathcal{CN}(0, \sigma_\alpha^2)$ and the independence of $\left\{\alpha_{l,q,i}\right\}_{i}$ as per Definition \ref{def:offgrid_dD_sparsity}.
This gives: 
\begin{equation}
    \label{eq:autocorr3}
    R_{h_{l,q}}(\tau) = N_D \,\sigma_\alpha^2 \, e^{-\imath 2\pi \frac{q}{N T_{\rm s}}\tau} \, \text{sinc}\left(\frac{\tau}{N T_{\rm s}}\right).
\end{equation}
Power spectral density (PSD) $P_{h_{l,q}}(f)\triangleq\mathcal{F}\{R_{h_{l,q}}(\tau)\}$ is thus 
\begin{align}
    \label{eq:PSD}
    P_{h_{l,q}}(f)=
    N_D \, \sigma_\alpha^2 \, N T_{\rm s} \,\text{rect}\left(\left(f-\frac{q}{N T_{\rm s}}\right)N T_{\rm s}\right).
\end{align}
The PSD of $h_{l,q}(t)$ thus satisfies the condition of Lemma \ref{lem:MSE_DPSS} with $B=\frac{1}{2N T_{\rm s}}$ and $F_{\rm c}=\frac{q}{N T_{\rm s}}$
giving
\begin{equation}
\label{eq:mse_dpss}
\begin{multlined}
    \mathbb{E}\left[\left|h_{l,q,n}-h_{l,q,n}^{\rm BEM}\right|^2\right]
    =\frac{1}{N}\mathbb{E}\left[\left\|\mathbf{h}_{l,q}-\mathbf{h}_{l,q}^{\rm BEM}\right\|^2\right]\\
    =\frac{N_{\rm D}\, \sigma_\alpha^2}{2WN}\sum_{b=Q_{\rm BEM}}^{N-1}\lambda_{b}^{(N,W)}\:.
\end{multlined}
\end{equation}
Now define $\lambda_b^{(c)}$ as the $b$-th eigenvalue of the prolate spheroidal wave functions (PSWF) \cite{PSWF} with the bandwidth parameter $c\triangleq\pi N W$. This allows us to exploit existing results on the behavior of PSWF eigenvalues in the limit of $c\to\frac{\pi}{2}$, i.e., as $W\to0$ at rate $\frac{1}{2N}$, to upper bound the sum of DPSS eigenvalues $\lambda_{b}^{(N,W)}$ in \eqref{eq:mse_dpss}. Indeed, due to \cite[Theorem 2]{DPSS_PSWF}
\begin{equation}
    \label{eq:lambda_b_ineq}
    \lambda_{b}^{(N,W)}\leq A_W\lambda_{b}^{(c)},\forall b=1,\ldots,N\:.
\end{equation}
Here, $A_W$ is a function of $W$ defined in \cite[Eq. (45)]{DPSS_PSWF} and its image is fully included in the interval $[\frac{\pi^2}{8},2]$. Plugging \eqref{eq:lambda_b_ineq} into \eqref{eq:mse_dpss} and noting that $2WN=1$ leads to
\begin{equation}
\label{eq:dpss_pswf}
    \mathbb{E}\left[\left|h_{l,q,n}-h_{l,q,n}^{\rm BEM}\right|^2\right]
    \leq N_{\rm D}\, \sigma_\alpha^2\, A_W\sum_{b=Q_{\rm BEM}}^{N}\lambda_{b}^{(c)}
\end{equation}
The right-hand side term in \eqref{eq:dpss_pswf} can be upper bounded due to the fact that the PSWF eigenvalues decay at least exponentially
\footnote{Actually, even super-geometric decay can be proven \cite{super_gemetric_lambda}} as $b$ grows beyond $\frac{2\pi N W}{\pi}+O(\log(\pi N W))=1+O(\log\frac{\pi}{2})$. More precisely, it follows from \cite[Theorem 2.5]{lambda_PSWF} that $\lambda_b^{(c)}=O(e^{-\frac{\pi^2}{\log\frac{\pi}{2}}b})$.
Plugging this into \eqref{eq:dpss_pswf} gives, for any $\epsilon>0$ and $Q_{\rm BEM}>C\log\frac{2}{\epsilon}$ for sufficiently large $C$,
\begin{equation}
\label{eq:upperbound_mse_lqn}
    \mathbb{E}\left[\left|h_{l,q,n}-h_{l,q,n}^{\rm BEM}\right|^2\right]<
    \frac{N_{\rm D}\,\sigma_\alpha^2\,A_W\epsilon}{2}\:.
\end{equation}
Now, note that
\begin{align}
\label{eq:total_mse}
    &\mathbb{E}\left[\sum_{l=0}^{L-1}\left|h_{l,n}-h_{l,n}^{\rm BEM}\right|^2\right]\nonumber\\
    &=\mathbb{E}\left[\sum_{l=0}^{L-1}\left|\sum_{q=-Q}^{Q} I_{l,q} e^{\imath 2\pi\frac{n q}{N}} \sum_{b=1}^{Q_{\rm BEM}}\left(h_{l,q,n}-h_{l,q,n}^{\rm BEM}\right)\right|^2\right]\nonumber\\
    &=\sum_{l=0}^{L-1}\sum_{q=-Q}^{Q}\mathbb{E}\left[I_{l,q}\right]\mathbb{E}\left[\left|h_{l,q,n}-h_{l,q,n}^{\rm BEM}\right|^2\right]
\end{align}
Plugging \eqref{eq:upperbound_mse_lqn} into the right-hand side of \eqref{eq:total_mse} gives
\begin{equation}
    \mathbb{E}\left[\sum_{l=0}^{L-1}\left|h_{l,n}-h_{l,n}^{\rm BEM}\right|^2\right]
    <\sum_{l,q}\mathbb{E}\left[I_{l,q}\right]\frac{N_{\rm D}\,\sigma_\alpha^2\,A_W\epsilon}{2}
    \leq\epsilon
\end{equation}
where the second inequality is due to the fact that $A_W\leq2$ and that $\sum_{l=0}^{L-1}\sum_{q=-Q}^{Q}\mathbb{E}\left[I_{l,q}\right]=\frac{1}{N_{\rm D}\,\sigma_\alpha^2}$ is due to the power normalization condition in \eqref{eq:power_normalization_off_grid}. This completes the proof.

\section{Proof of Theorem \ref{theo:prediction}}
\label{app:prediction_proof}
We now use the notation $h_{l,q,n}^{\rm ext}\left(\hat{\mathbf{h}}_{l,q}^{\rm BEM}\right)$ for the DPSS extrapolation predictor defined in \eqref{eq:h_lqn_ext_u} to highlight its dependence on the estimate $\hat{\mathbf{h}}_{l,q}^{\rm BEM}$. When $\hat{\mathbf{h}}_{l,q}^{\rm BEM}$ in \eqref{eq:h_lqn_ext_u} is replaced with an arbitrary vector $\mathbf{h}$, the DPSS predictor generalizes to
\begin{equation}
\label{eq:h_lqn_ext_h}
    h_{l,q,n}^{\rm ext}\left(\mathbf{h}\right)\triangleq \underbrace{e^{\imath 2\pi\frac{n q}{N}}\left(\mathbf{u}_n^{\rm ext}\right)^{\rm T}\mathbf{U}_{Q_{\rm BEM}}^{\rm H}\mathbf{E}_{\frac{q}{N}}^{\rm H}}_{\triangleq \left(\mathbf{f}_{Q_{\rm BEM}}^{\rm ext}\right)^{\rm H}}\mathbf{h}\:.
\end{equation}
For instance, $h_{l,q,n}^{\rm ext}\left(\mathbf{h}_{l,q}^{\rm BEM}\right)$ is the DPSS predictor based on the actual, not the estimated, vector $\mathbf{h}_{l,q}^{\rm BEM}$. Next, we establish the link between the DPSS extrapolation predictor and an MMSE predictor that uses the knowledge of the channel values during the observation interval.
For that sake, first note that for each $l$, $q$ and any $n\in\mathbb{Z}$, the random variable $h_{l,q,n}$ as defined by \eqref{eq:h_lqn} follows a complex symmetric Gaussian distribution $\mathcal{CN}\left(0,N_{\rm D}\sigma_\alpha^2\right)$ under the conditions of Definition \ref{def:offgrid_dD_sparsity}. Moreover, the random process $\left(h_{l,q,n}\right)_{n\in\mathbf{Z}}$ is stationary and has an auto-correlation $\mathbb{E}[h_{l,q,n}h_{l,q,m}^*]=\sigma_{\alpha}^2 N_{\rm D} e^{\imath 2\pi\frac{(n-m)q}{N}} \frac{N}{\pi (n-m)}\sin{(\frac{\pi (n-m)}{N})}$ due to \eqref{eq:h_lqn}.
Similarly, $\mathbf{h}_{l,q}\sim\mathcal{CN}(\mathbf{0},\sigma_{\alpha}^2 N_{\rm D}N\mathbf{E}_{\frac{q}{N}}\boldsymbol{\Sigma}\mathbf{E}_{\frac{q}{N}}^{\rm H})$ and $\mathbf{h}_{l,q}^{\rm BEM}\sim\mathcal{CN}(\mathbf{0},\sigma_{\alpha}^2 N_{\rm D}N\mathbf{E}_{\frac{q}{N}}\mathbf{P}^{\rm BEM}\boldsymbol{\Sigma}\mathbf{P}^{\rm BEM}\mathbf{E}_{\frac{q}{N}}^{\rm H})$ with 
\begin{equation}
    \label{eq:sigma}
    \boldsymbol{\Sigma}\triangleq\Bigg[\underbrace{\frac{1}{\pi (n-m)} \sin{\frac{\pi(n-m)}{N}}}_{C_{n,m}^{(N,W)}}\Bigg]_{\substack{n=0\ldots N-1 \\ m=0\ldots N-1}}
\end{equation}
Therefore, the MMSE predictor of $h_{l,q,n}$ given the actual channel component $\mathbf{h}_{l,q}$ is
\begin{align}
    \hat{h}_{l,q,n}\left(\mathbf{h}_{l,q}\right)&=\mathbf{E}\left[h_{l,q,n}\mathbf{h}_{l,q}^{\rm H}\right]\left(\mathbf{E}\left[\mathbf{h}_{l,q}\mathbf{h}_{l,q}^{\rm H}\right]\right)^{-1}\mathbf{h}_{l,q}\nonumber\\
    &=e^{\imath 2 \pi \frac{nq}{N}}\boldsymbol{\rho}\mathbf{E}_{\frac{q}{N}}\boldsymbol{\Sigma}^{-1}\mathbf{E}_{\frac{q}{N}}^{\rm H}\mathbf{h}_{l,q}\:.
\end{align}
Here, $\boldsymbol{\rho}\triangleq[e^{-\imath 2 \pi \frac{mq}{N}} C_{n,m}^{(N,W)}]_{m=0}^{N-1}$.
It follows that the reduced-rank MMSE predictor $\hat{h}_{l,q,n}^{\rm RR}$ of rank $Q$ of $h_{l,q,n}$ given $\mathbf{h}_{l,q}$ (where notation `$\mathrm{RR}$' stands for ``reduced rank'') is \cite{reduced_rank}
\begin{equation}
    \hat{h}_{l,q,n}^{\rm RR}\left(\mathbf{h}_{l,q}\right)=\left(\mathbf{f}_Q^{\rm RR}\right)^{\rm H}\mathbf{h}_{l,q}\:,
\end{equation}
\begin{equation}
\label{eq:f_q_RR}
    \mathbf{f}_Q^{\rm RR}\triangleq\mathbf{E}_{\frac{q}{N}}\mathbf{U}_{Q}\mathrm{diag}(\frac{1}{\lambda_1^{(N,W)}},\ldots,\frac{1}{\lambda_Q^{(N,W)}})\mathbf{U}_{Q}^{\rm H}\mathbf{E}_{\frac{q}{N}}^{\rm H}\boldsymbol{\rho}^{\rm H}.
\end{equation}
Next, we apply the triangle inequality to get
\begin{equation}
    \begin{multlined}
        \mathbb{E}\left[\left|h_{l,q,n}^{\rm ext}\left(\hat{\mathbf{h}}_{l,q}^{\rm BEM}\right)-\hat{h}_{l,q,n}^{\rm RR}\left(\mathbf{h}_{l,q}\right)\right|^2\right]\leq\\
        \underbrace{\mathbb{E}\left[\left|h_{l,q,n}^{\rm ext}\left(\hat{\mathbf{h}}_{l,q}^{\rm BEM}\right)-h_{l,q,n}^{\rm ext}\left(\mathbf{h}_{l,q}^{\rm BEM}\right)\right|^2\right]}_{\triangleq E_1}+\\
        \underbrace{\mathbb{E}\left[\left|h_{l,q,n}^{\rm ext}\left(\mathbf{h}_{l,q}^{\rm BEM}\right)-h_{l,q,n}^{\rm ext}\left(\mathbf{h}_{l,q}\right)\right|^2\right]}_{\triangleq E_2}+\\
        \underbrace{\mathbb{E}\left[\left|h_{l,q,n}^{\rm ext}\left(\mathbf{h}_{l,q}\right)-\hat{h}_{l,q,n}^{\rm RR}\left(\mathbf{h}_{l,q}\right)\right|^2\right]}_{\triangleq E_3}\:.
    \end{multlined}
\end{equation}
Since $\lim_{\sigma_w^2\to0}\mathbb{E}[\|\hat{\mathbf{h}}_{l,q}^{\rm BEM}-\mathbf{h}_{l,q}^{\rm BEM}\|^2]$ due to Assumption \ref{assum:Np}, it follows from \eqref{eq:h_lqn_ext_h} by standard MSE derivations that $\lim_{\sigma_w^2\to0}E_1=0$.
As for $E_2$, and since $\mathbf{h}_{l,q}^{\rm BEM}=\mathbf{E}_{\frac{q}{N}}\mathbf{P}^{\rm BEM}\mathbf{E}_{\frac{q}{N}}^{\rm H}\mathbf{h}_{l,q}$ per \eqref{eq:BEM2}, we have that $\mathbf{U}_{Q_{\rm BEM}}^{\rm H}\mathbf{E}_{\frac{q}{N}}^{\rm H}\mathbf{h}_{l,q}=\mathbf{U}_{Q_{\rm BEM}}^{\rm H}\mathbf{E}_{\frac{q}{N}}^{\rm H}\mathbf{h}_{l,q}^{\rm BEM}$. It follows from \eqref{eq:h_lqn_ext_h} that $h_{l,q,n}^{\rm ext}\left(\mathbf{h}_{l,q}\right)=h_{l,q,n}^{\rm ext}\left(\mathbf{h}_{l,q}^{\rm BEM}\right)$ and hence that $E_2=0$. Finally, note by referring to \eqref{eq:dpss_ch_predictor}, \eqref{eq:h_lqn_ext_h} and \eqref{eq:f_q_RR} that $\mathbf{f}_{Q}^{\rm RR}=\mathbf{f}_{Q_{\rm BEM}}^{\rm ext}$ if $Q=Q_{\rm BEM}$ leading to $E_3=0$.
Now, due to \eqref{eq:dpss_ch_predictor}, $h_{l,n}^{\rm ext}=\sum_{l=0}^{L-1}I_{l,q}h_{l,q,n}^{\rm ext}$. Moreover, due to the independence conditions from Definition \ref{def:offgrid_dD_sparsity}, $\hat{h}_{l,n}^{\rm RR}\triangleq\sum_{l=0}^{L-1}I_{l,q}\hat{h}_{l,q,n}^{\rm RR}$ is the reduced-rank MMSE estimate of $h_{l,n}$ given $\left\{\mathbf{h}_{l,q}\right\}_{q=-Q\cdots Q}$ and conditioned on a given realization of $I_{l,q}$.
Putting all these pieces together, it follows that $\lim_{\sigma_w^2\to0}|h_{l,n}^{\rm ext}-\hat{h}_{l,n}^{\rm RR}|^2$. This completes the proof.



\end{document}